\documentclass[numbers]{COMNET}
\bibpunct{[}{]}{,}{n}{,}{;}
\usepackage{latexsym}
\usepackage[T1]{fontenc}
\usepackage{a4wide}
\usepackage{amsmath}
\usepackage{graphicx}
\usepackage{epsfig}
\usepackage{dsfont}
\usepackage{color,float}
\usepackage{amsthm}
\usepackage{url}
\newtheorem{theo}{Theorem}[section]

\newtheorem{prop}[theo]{Proposition}

\theoremstyle{remark}
\newtheorem{rem}[theo]{Remark}
\theoremstyle{remark}
\newtheorem{ex}[theo]{Example}

\newcommand{\vc}[1]{\boldsymbol{#1}}

\newcommand{\diag}{\textrm{diag}}

\definecolor{darkred}{rgb}{0.6,0,0}
\usepackage{color}

\usepackage{amssymb}

\begin{document}

\title{Sensitivity analysis of a branching process evolving on a network with application in epidemiology}

\shorttitle{Sensitivity analysis of a branching process evolving on a network} 
\shortauthorlist{S. Hautphenne, G. Krings, J.-C. Delvenne, and V. D. Blondel} 

\author{
\name{Sophie Hautphenne$^*$}
\address{University of Melbourne, Department of Mathematics and Statistics, Melbourne, VIC 3010, Australia. \email{$^*$Corresponding author: sophiemh@unimelb.edu.au}}
\name{Gautier Krings}
\address{Institute for Information and Communication Technologies, Electronics and Applied
Mathematics (ICTEAM), Universit\'e catholique de Louvain, Avenue Georges Lema\^itre, 4,
1348 Louvain-la-Neuve, Belgium}
\name{Jean-Charles Delvenne}
\address{Institute for Information and Communication Technologies, Electronics and Applied
Mathematics (ICTEAM), Universit\'e catholique de Louvain, Avenue Georges Lema\^itre, 4,
1348 Louvain-la-Neuve, Belgium\\Center for Operation Research and Econometrics (CORE), Universit\'e catholique de Louvain, Voie du Roman Pays 34, 1348 Louvain-la-Neuve, Belgium}
\and
\name{Vincent D. Blondel}
\address{Institute for Information and Communication Technologies, Electronics and Applied
Mathematics (ICTEAM), Universit\'e catholique de Louvain, Avenue Georges Lema\^itre, 4,
1348 Louvain-la-Neuve, Belgium}}

\maketitle

\begin{abstract}
{We perform an analytical sensitivity analysis for a model of a continuous-time branching process evolving on a fixed network. This allows us to determine the relative importance of the model parameters to the growth of the population on the network. We then apply our results to the early stages of an influenza-like epidemic spreading among a set of cities connected by air routes in the United States. 
We also consider vaccination and analyze the sensitivity of the total size of the epidemic with respect to the fraction of vaccinated people. Our analysis shows that the epidemic growth is more sensitive with respect to transmission rates within cities than travel rates between cities. More generally, we highlight the fact that branching processes offer a powerful stochastic modeling tool with analytical formulas for sensitivity which are easy to use in practice.}
 
 
 {branching process; diffusion on network; sensitivity analysis; influenza; epidemic model}
\end{abstract}


\section{Introduction}

Branching processes are powerful stochastic models that describe the evolution of populations of individuals which reproduce and die independently of each other according to specific probability laws. They are playing an increasingly important role in models of population biology including molecular biology, ecology, epidemiology, and evolutionary theory \cite{axelrod2002branching}, as well as in other scientific areas such as particle physics, chemistry, and computer science \cite{haccou2005branching}. Typical performance measures of these models include the distribution of the instantaneous and cumulative population sizes at a given time, the extinction probability of a population, and its asymptotic growth rate and composition.

 In order to apply branching processes to real-world problems, the parameters of the model must be estimated from available data sets. Small errors or changes in the parameters may lead to notably different model outputs. A sensitivity analysis may quantify the impact of each parameter on the performance measures of the model. This may be useful if we want to know which parameters influence the growth of the branching process the most. To the best of our knowledge, sensitivity analysis has received scant attention in the branching processes literature. Apart from a few papers dealing with very specific questions such as \cite{caron2006branching} and \cite{fox2005extinction}, there is no complete study of the topic for more general classes of branching processes. 
 
 Here we consider a continuous-time Markovian branching process evolving on a fixed network. Such a process belongs to the class of \textit{Markovian trees}, which can be seen as particular multitype branching processes, and offer considerable modelling versatility and appealing computational properties, see \cite{bkt08}, \cite{hautphenne2012markovian} and \cite{hautphenne2011algorithmic}. The concept of a branching process evolving on a network is certainly not new, but our aim in this paper is to present 
 branching processes and their properties in an accessible way, and to show how to derive practical, analytic sensitivity formulas for the main quantities of interest. For that purpose, we first discuss the typical performance measures of the model (as listed earlier), and we then perform an analytical sensitivity analysis of each performance measure with respect to the model parameters. As opposed to the performance measures themselves, most of the sensitivity formulas have an explicit expression that can be easily used in practical situations. 
 
We then illustrate our theoretical sensitivity results on a topical real-world problem: the spread of an influenza-like epidemic on a network of cities in the United States through air traffic. It is well known that the early stages of an epidemic can be approximated by a branching process \cite{bd93}, \cite{oneill95}. The averaged branching process is essentially the linearisation of the nonlinear \textit{Susceptible-Infected-Removed} (\textit{SIR}) compartmental model around the disease-free equilibrium. We refer to Hethcote \cite{heth00} and the book of Keeling and Rohani \cite{keeling2008modeling} for a good introduction to the modeling of infectious diseases, to Arino \cite{arino2009diseases} for an overview on diseases in metapopulations, and to Balcan et al. \cite{balcan10} for a recent computational model for the spatial spread of infectious diseases.

The sources of errors in the outputs of a model of epidemic evolving on a network are manifold. They may arise either from the simplifying assumptions of the model, which necessarily neglect many features of the real world, or from imprecise epidemiological or mobility data. 
The former is addressed in the most sophisticated models by taking into account a large number of explanatory variables or compartments, such as age, gender, location, etc., and modeling the dynamics in a complex nonlinear way. The mobility data, at a global level or at the level of a large country such as the United States, is often estimated from the number of passengers flying from one airport to the other. These data are easily available and are assumed to account for a large part of the mobility. 
A discussion on the relevance of these data is included in \cite{balcan09}. The epidemiological data are however much harder to estimate accurately. For instance, the average number of infectious contacts that an individual infected by seasonal influenza makes in a day have been estimated to range from $0.55$ to $1.44$, see \cite{gekg04}. 
  

Sensitivity of epidemic models to parameters has already been studied for different types of diseases: Hyman and LaForce \cite{hyman2003modeling} studied the sensitivity of a multi-city deterministic epidemic model for influenza with respect to the epidemiological data. They assumed that the transmission and recovery rates are the same in each city and show that the recovery rate is the most important single parameter. Chitnis \textit{et al.} \cite{chitnis2008determining} perform a sensitivity analysis on a non-linear compartmental model of malaria transmission, and obtain an analytical formula for the basic reproduction number $R_0$ and for its sensitivity.

In this paper, we compute the sensitivity of the branching process approximation with respect to both epidemiological and domestic mobility data.
We find that the results are significantly more sensitive with respect to the epidemiological data than with respect to the domestic mobility data. This provides confidence in the use of the general approximation of domestic mobility by airport traffic, but suggests that the most stringent limitation of current models is in the relatively imprecise epidemiological data, rather than in the sophistication of the model or the estimation of domestic traffic. This is in contrast with the role of international air travel, which can have a more important impact on the development of some diseases such as H1N1 \cite{khan2009spread}. 
 
Finally, we refine our sensitivity analysis with respect to the epidemiological parameters by studying the effect of a vaccination campaign. 
 The sensitivity analysis of the epidemic size with respect to the vaccination rate allows us to compute a practically relevant quantity: the impact of each supplementary vaccination shot on the total number of infected people, when the vaccination has succeeded in stopping the exponential spreading of the disease. In this regime the stochastic branching process approximation is valid at all times \cite{bd93}.
 
In an appendix, we also introduce the possibility of on-board transmission. Infection during transportation has been analyzed from a theoretical point of view for the \emph{SIR} model in \cite{knipl2012multiregional}, but the authors have not applied the model to real situations.  Studies on the transmission rates on board airplanes exist (see for instance the report of the European Centre for Disease Prevention and Control \cite{ecdc}), however, to the best of our knowledge, data are not sufficient to provide accurate on-board transmission functions for influenza. Therefore we assume that the transmission probabilities are given by a specific function of the flight time, and we show how this assumption affects the sensitivity results. We show in particular that considering on-board transmission can noticeably modify the shape of the sensitivity curves.
 
 The paper is organised as follows. In Section 2, we describe the model of branching process evolving on a network and we discuss several performance measures which can be computed from the model. In Section 3, we define the sensitivity and the elasticity of a model output with respect to a parameter, and we derive analytical expressions for the sensitivity of each performance measure discussed in Section 2. We end the section by an illustration of the sensitivity results on a simple artificial example. In Section 4, we apply the sensitivity results on a model for the early spread of an influenza-like epidemic on a network of cities in the United States. Finally, we introduce vaccination and perform a sensitivity analysis with respect to the proportion of vaccinated people. We conclude our paper in Section 5. The technical details and supplementary information are provided in six appendices. 
 
 Throughout the paper, column vectors are denoted by $\vc x$ and row vectors are denoted by $\vc x^\top$. The column vector $\vc e_i$ denotes the unit vector with a 1 at the $i$th entry and zeros elsewhere, and $\vc 1$ and $\vc 0$ are column vectors of which all elements are respectively equal to one and zero; the size of these vectors is generally clear from the context.

\section{A model of branching process evolving on a network} \label{bpn}

We consider a fixed network represented by an undirected graph $(\mathcal{G},\mathcal{E})$ where $\mathcal{G}=\{1,2,\ldots,n\}$ is the set of nodes and $\mathcal{E}$ is the set of vertices. Two nodes $i$ and $j$ are \textit{adjacent} when either $i=j$ or the edge $(i,j)$ belongs to $\mathcal{E}$. 

We define a continuous-time stochastic model of a population evolving on the network as follows. The process starts at time 0 with a single individual located at node $i\in\mathcal{G}$. One of the following three types of event may then happen to this individual while at node $i$: \begin{itemize}
\item The individual moves from $i$ to an adjacent node $j\neq i$; this happens at rate $T_{ij}$ (that is, $T_{ij}$ transitions from $i$ to $j$ occur on average  per time unit, per individual at node $i$).
\item The individual gives birth to $k$ children ($k\geq 1$) {and} simultaneously moves to node $j$ adjacent to $i$, the $k$ children respectively starting
their life at nodes $j_1$, $j_2,\ldots,j_k$ adjacent to $i$; this happens at rate $(B_{k})_{i;j_1\,j_2\ldots\,j_k\,j}$. 
\item The individual dies; this happens at rate $d_i$.
\end{itemize} All individuals living on the network behave independently of each other with the same rules as the initial individual. The transition rates, birth rates and death rates are gathered respectively in a $n \times n$ matrix $T=(T_{ij})$, a sequence of $n \times n^{k+1}$ matrices $B_{k}=((B_{k})_{i;j_1\,j_2\ldots\,j_k\,j})$,  $k\geq 1$, and an $n\times 1$ vector $\vc d=(d_i)$. The diagonal elements of $T$ are strictly negative and $|T_{ii}|$ is the parameter of the exponential distribution of the sojourn time of an individual at node $i$ before one of the three abovementioned events occurs. 
These elements are computed such that the matrices and vector satisfy $T\vc 1 + \sum_{k\geq1} (B_k \vc 1)+\vc d=\vc 0$.
The resulting Markovian population process is a branching process characterized by the set of matrices $\{T,\{B_k\}_{k\geq1},\vc d\}$ belonging to the class of \textit{Markovian trees}, see
\cite{bkt08} and \cite{soph5}. 

\begin{ex} We illustrate the structure of the matrices $T,\{B_k\}_{k\geq1}$ and $\vc d$ on the simplest network with $n=2$ (adjacent) nodes. The $2\times 2$ transition rate matrix $T$ is then given by
$$T=\left[\begin{array}{cc} T_{11}&T_{12}\\T_{21} & T_{22}\end{array}\right],$$ where the diagonal entries $T_{ii}$ will be described further.
Assume that the individuals of the branching process can give birth to at most two children at each birth event. This means that the birth rate matrices $B_k$ are nonzero only for $k=1$ and $k=2$. The $2\times 4$ matrix $B_1$ has the following structure
$$B_1=\left[\begin{array}{cccc} (B_1)_{1,11}& (B_1)_{1,12}& (B_1)_{1,21}& (B_1)_{1,22}\\(B_1)_{2,11}& (B_1)_{2,12}& (B_1)_{2,21}& (B_1)_{2,22}\end{array}\right],$$ and the $2\times 8$ matrix $B_2$ has the following structure:
$$B_2=\left[\begin{array}{cccccccc} (B_2)_{1,111}& (B_2)_{1,112}& (B_2)_{1,121}& (B_2)_{1,122}&(B_2)_{1,211}& (B_2)_{1,212}& (B_2)_{1,221}& (B_2)_{1,222}\\(B_2)_{2,111}& (B_2)_{2,112}& (B_2)_{2,121}& (B_2)_{2,122}&(B_2)_{2,211}& (B_2)_{2,212}& (B_2)_{2,221}& (B_2)_{2,222}\end{array}\right].$$ For instance, the entry $(B_2)_{1,122}$ is the rate at which a parent at node 1 gives birth to two children and instantaneously moves to node 2, while one of his children stays at node 1 and the second one moves to node 2. Note that the indices in the entries of $B_1$ and $B_2$ are ordered lexicographically by convention.
Finally the $2\times 1$ death rate vector is $\vc d=[d_1,d_2]^ \top$. In order for $T\vc 1 + B_1 \vc 1+B_2\vc 1+\vc d=\vc 0$ to hold, the diagonal entries of the matrix $T$ are then given by
\begin{eqnarray*}T_{11}&=&-d_1-T_{12}-\sum_{j_1,j=1}^2( B_1)_{1,j_1j}-\sum_{j_1,j_2,j=1}^2( B_2)_{1,j_1 j_2j},\\T_{22}&=&-d_2-T_{21}-\sum_{j_1,j=1}^2( B_1)_{2,j_1j}-\sum_{j_1,j_2,j=1}^2( B_2)_{2,j_1 j_2j}.\end{eqnarray*}
\end{ex}
%

In the next subsections, we describe several performance measures that can be computed for a branching process evolving on a network. For the clarity of the presentation, the performance measures will be discussed first in the scalar case ($n=1$) which corresponds to the standard Markov branching process \cite[Chapter V]{harris2002theory}, and next in the matrix case for an arbitrary number of nodes $n>1$. Note that in the scalar case, the matrices $T,\{B_k\}_{k\geq1}$ and $\vc d$ become all scalar: $B_k$ is the rate at which an individual gives birth to $k$ children, $d$ is the death rate, and $T=-d-\sum_k B_k$. 

The matrix case will require the use of the Kronecker product between matrices, which is defined as follows: for an $n\times m$ matrix $A$ and a $p\times q$ matrix $B$, the Kronecker product of $A$ and $B$, denoted by $A\otimes B$, is an $n p\times mq$ matrix defined as 
$$A\otimes B=\left[\begin{array}{cccc} A_{11} B & A_{12}B&\ldots & A_{1m}B\\A_{21} B & A_{22}B&\ldots & A_{2m}B\\\vdots & & &\vdots\\A_{n1} B & A_{n2}B&\ldots & A_{nm}B\end{array}\right].$$ It will also require the knowledge of the Perron-Frobenius Theorem for nonnegative matrices, see \cite{seneta2006non}.

Most of the results in the scalar case can be found in \cite{harris2002theory}, and most of the results in the matrix case can be found in \cite{hautphenne2012markovian} and \cite{soph5} for the case where $B_k=0$ for $k\geq 2$ (the Markovian \textit{binary} tree case).

\subsection{Instantaneous population size}

\paragraph{Scalar case $n=1$.}  Let $Z(t)$ denote the population size in the branching process at time $t$. The process $\{Z(t),t\geq 0\}$ is a continuous-time Markov chain on the nonnegative integers, where state 0 is absorbing and all other states are transient (that is, they will be visited a finite number of times with probability one). The so-called generating function of $Z(t)$ is defined as 
$Q(t,s)=\sum_{k\geq 0} P[Z(t)=k]s^k$, where $|s|\leq 1$, and is well known to satisfy the backward 
  Kolmogorov differential equation
\begin{eqnarray} \label{equdiffscal}\frac{\partial Q(t, s)}{\partial t}&= &d+T\, Q(t,s)+\sum_{k\geq 1} B_k\,Q(t,s)^{k+1}, \end{eqnarray} with $Q(0, s)=s$ if the population starts with a single individual at time 0 (which we assume here).

All the moments of $Z(t)$ can be obtained by successive derivatives of $Q(t, s)$ at $s=1$. In particular, the mean population size at time $t$, $M(t)=E[Z(t)]$, is given by $M(t)=(\partial Q(t,s)/\partial s)|_{s=1}.$ By differentiating \eqref{equdiffscal} with respect to $s$, at $s=1$, we obtain the linear differential equation for $M(t)$
\begin{equation}\label{MMscal}\frac{\partial M(t)}{\partial t}=\Omega\,M(t),\end{equation}
 with $M(0)=1$, where \begin{equation}\label{omdefscal}\Omega = T +\sum_{k\geq 1}\,B_k\,(k+1)=-d+\sum_{k\geq 1}\,k\,B_k.\end{equation}
 By solving \eqref{MMscal} we obtain that the mean population size at time $t$ is given by
\begin{equation} \label{Mtscal}M(t)=\exp(\Omega t).\end{equation}
\paragraph{Matrix case $n>1$.} We now extend the above definitions and results to the $n$-dimensional case. 
The vector $\vc Z(t)=[Z_1(t),\ldots,Z_n(t)]$ records the population size at time $t$ at each of the $n$ nodes. The evolution of the branching process now depends on the node occupied by the initial individual at time 0. We therefore define the \textit{conditional} probability generating function of $\vc Z(t)$, given that the process starts with a first individual at node $i$, as$$Q_i(t,\vc s)=\sum_{\vc k\geq\vc 0} P[\vc Z(t)=\vc k | \vc Z(0)=\vc e_i] \vc s^{\vc k},$$ where $\vc k=(k_1,\ldots,k_n)\in \mathbb{N}^n$, $\vc s=(s_1,\ldots,s_n)$, $|s_i|\leq 1$, and $\vc s^{\vc k}:=s_1^{k_1}\cdots s_n^{k_n}$. 
As shown in \cite{soph5}, the vector function $\vc Q (t,\vc s)=(Q_i(t,\vc s))$ satisfies the matrix analogue of \eqref{equdiffscal},
\begin{eqnarray} \label{equdiff}\frac{\partial \vc Q(t,\vc s)}{\partial t}&= &\vc d+T\, \vc Q(t,\vc s)+\sum_k B_k\,\vc Q(t,\vc s)^{(k+1)}, \end{eqnarray} with $\vc Q(0,\vc s)=\vc s$, where
$\vc Q(t,\vc s)^{(k+1)} $ stands for the $(k+1)$st-fold Kronecker
product of the vector $\vc Q(t,\vc s)$ with itself: $\vc Q(t,\vc s)^{(0)} = 1,$ and
$\vc Q(t,\vc s)^{(k)}= {\vc Q(t,\vc s)}^{(k-1)} \otimes \vc Q(t,\vc s)$, for $ k\geq 1$. 

The mean population size at time $t$ is now given by a matrix $M(t)=(M_{ij}(t))$, where $M_{ij}(t)=E[Z_j(t)|\vc Z(0)=\vc e_i]$ is the (conditional) mean number of individuals at node $j$ in the population at time $t$, given that the process started at time $0$ with one individual at node $i$. The entries of $M(t)$ being obtained as $M_{ij}(t)=(\partial Q_i(t,\vc s)/\partial s_j)|_{\vc s=\vc1}$, the same differential equation \eqref{MMscal} and solution \eqref{Mtscal} hold for the matrix $M(t)$, the only difference being that $M(0)=I$, and $\Omega$ is now a matrix given by
 \begin{equation}\label{omdef}\Omega = T +\sum_{k\geq 1}\,B_k\,\sum_{i=0}^k( \vc 1^{(i)}\otimes I \otimes \vc 1^{(k-i)}).\end{equation} Note that the matrix exponential is defined as $$e^{\Omega t}=\sum_{n\geq0} \dfrac{\Omega^n t^n}{n!}.$$

\subsection{Cumulative population size}
\paragraph{Scalar case $n=1$.}Let $N(t)$ be the cumulative number of individuals born in the branching process until time $t$, and let $G(t,s)$ be the probability generating function of $N(t)$. Similar to $Q(t, s)$, $G(t,s)$ satisfies the differential equation \begin{eqnarray} \label{eqdiff2scal}\frac{\partial  G(t,s)}{\partial t}&= & d\,s+T\,  G(t,s)+\sum_k B_k\, G(t,s)^{k+1}, \end{eqnarray} with $G(0,s)=s$. Note that the only difference with Eq. \eqref{equdiffscal} is that here $d$ is multiplied by $s$.  

Again, all the moments of $N(t)$ can be obtained by successive derivatives of $G(t,s)$ at $s=1$, and we focus here on the mean cumulative number of individuals born in the branching process until time $t$, $D(t)=E[N(t)]=(\partial G(t,s)/\partial s)|_{s=1}$. By differentiating Eq. \eqref{eqdiff2scal} with respect to $s$, at $s=1$, we obtain the following differential equation for $D(t)$:
 $$\frac{\partial D(t)}{\partial t}=d+\Omega\, D(t),$$ with $ D(0)=1$ and where the scalar $\Omega$ is defined in \eqref{omdefscal}. The mean cumulative population size until time $t$ is thus given by
\begin{equation}\label{Dtscal} D(t)=[I-\exp(\Omega t)]\,(-\Omega)^{-1}\,d+\exp(\Omega t).\end{equation} 
\paragraph{Matrix case $n>1$.}
In the $n$-dimensional case, we extend the definition of $G(t,s)$ to its vector analogue: we define the conditional probability generating function of $N(t)$, given that the process started with a first individual at node $i$, by
$$G_i(t,s)=\sum_{k\geq0} P[N(t)=k|\vc Z(0)=\vc e_i]\,s^k.$$ Note that here, $s$ is still scalar. The matrix analogue of Eq. \eqref{eqdiff2scal} for the vector $\vc G(t,s)=(G_i(t,s))$ is \begin{eqnarray} \label{eqdiff2}\frac{\partial \vc G(t,s)}{\partial t}&= &\vc d\,s+T\, \vc G(t,s)+\sum_k B_k\,\vc G(t,s)^{(k+1)}, \end{eqnarray} with $\vc G(0,s)=s\,\vc1$, see \cite{soph5}. 
Since we now condition on the node of the initial individual at time 0, the mean cumulative number of individuals born in the population until time $t$ becomes a vector $\vc D(t)=(D_i(t))$ where
$D_i(t)=E[N(t)|\vc Z(0)=\vc e_i]=(\partial G_i(t,s)/\partial s)|_{s =1}$. Using the same technique as in the scalar case, we find that the vector $\vc D(t)$ is given by the matrix analogue of \eqref{Dtscal},
\begin{equation}\label{Dt} \vc D(t)=[I-\exp(\Omega t)]\,(-\Omega)^{-1}\,\vc d+\exp(\Omega t)\,\vc 1,\end{equation}where the matrix $\Omega$ is defined in \eqref{omdef}. Note that this expression requires $\Omega$ to be nonsingular, which is generally the case in most practical situations.

\subsection{Extinction probability}
\paragraph{Scalar case $n=1$.} Due to the transient nature of the strictly positive states of $\{Z(t)\}$, the population in the branching process either eventually grows without bound or becomes extinct, there is no stationary behavior (but for trivial cases). We denote by $q$ the probability of eventual extinction of the process, that is, $q=\lim_{t\rightarrow\infty} P[Z(t)= 0]$. Since $P[Z(t)= 0]=Q(t, 0)$, an equation for $q$ can be obtained by taking $s=0$ and $t\rightarrow\infty$ in Eq. \eqref{equdiffscal}:
\begin{eqnarray} \label{exteqscal}0&= & d+T\,  q+\sum_k B_k\, q^{k+1}. \end{eqnarray} This non-linear equation has potentially more than one solution, but it can be shown that $q$ is its minimal nonnegative solution.
\paragraph{Matrix case $n>1$.} The $n$-dimensional version of $q$ is the vector $\vc q=(q_i)$ where
$q_i=\lim_{t\rightarrow\infty} P[\vc Z(t)=\vc 0|\vc Z(0)=\vc e_i]$ is the conditional extinction probability of the branching process, given that it starts with one individual at node $i$. Similar to the scalar case, we can show that the vector $\vc q$ is the (componentwise) minimal nonnegative solution of the matrix analogue of \eqref{exteqscal},
\begin{eqnarray} \label{exteq}\vc0&= &\vc d+T\, \vc q+\sum_k B_k\,\vc q^{(k+1)}. \end{eqnarray}Several algorithms with a probabilistic interpretation have been developed to solve for $\vc q$, see \cite{bkt08}, \cite{hautphenne2008newton}, and \cite{hautphenne2011algorithmic}.

\subsection{Extinction criteria} \label{extcrit}
\paragraph{Scalar case $n=1$.} Various criteria can be used to determine whether the population eventually becomes extinct with probability one or has a positive probability of growing without bound, that is, whether $q=1$ or $ q<1$. We shall consider two equivalent extinction criteria here. 

First observe that, by \eqref{Mtscal}, the mean population size $M(t)$ has different asymptotic behavior depending on the sign of $\Omega$ defined in \eqref{omdefscal}: as $t\rightarrow\infty$, $M(t)\rightarrow 0$ if $\Omega<0$ (subcritical case), $M(t)=1$ if $\Omega=0$ (critical case), or $M(t)\rightarrow \infty$ if $\Omega>0$ (supercritical case). It follows that the population eventually becomes extinct with probability one ($q=1$) if and only if $\Omega\leq 0$.

Another threshold quantity is given by the mean number of children generated by an individual during its entire lifetime, that we denote by $R$. The population has a positive probability of surviving ($q<1$), if and only if $R>1$, that is, if on average, every individual is replaced by more than one individual. In the scalar case, $R$ is the ratio of the total rate at which an individual generates a child to the death rate, that is,\begin{equation}\label{exRR}R=\dfrac{\sum_k k B_k}{d}.\end{equation} Note that the two criteria are indeed equivalent since $\Omega\leq 0\Leftrightarrow R\leq 1$.
\paragraph{Matrix case $n>1$.} The two extinction criteria described in the scalar case have a counterpart in the $n$-dimensional case. Since $\Omega$ is now a matrix (as defined in \eqref{omdef}), the condition $\Omega\leq 0$ for almost sure extinction needs to be adapted: the eigenvalue $\Omega_0$ of maximal real part of $\Omega$ (also called the Perron-Frobenius eigenvalue) now plays the role of the threshold quantity:
$$\vc q=\vc 1 \quad \Leftrightarrow \quad \Omega_0\leq 0.$$


We define the matrix $R=(R_{ij})$, where the entry $R_{ij}$ is the expected total number of children born at node $j$ from a parent born at node $i$, during the entire lifetime of the parent. The second extinction criterion then relies on the eigenvalue $R_0$ of maximal real part of the matrix $R$: $$\vc q=\vc 1 \quad \Leftrightarrow \quad R_0\leq 1.$$
An explicit expression for $R$, which generalizes \eqref{exRR} to the matrix case, is given in Proposition~\ref{matR} in Appendix~\ref{apfurther}.

The matrix $R$ and its dominant eigenvalue $R_0$ will play a fundamental role in the epidemic application, as shown in Section~\ref{app}.

\subsection{Asymptotic node frequency}\label{freq}
It may be interesting to know what proportion of the inviduals living in the network at time $t$ are located at node $i$. We also call this proportion the frequency of node $i$ at time $t$, and when $t\rightarrow\infty$ we talk about the asymptotic node frequency. Note that these notions make sense only when the network has at least two nodes; therefore, we only consider the matrix case in this section.

The matrix exponential $\exp(\Omega)$ has the Perron-Frobenius eigenvalue $\exp(\Omega_0)$. Hence, as a consequence of (\ref{Mtscal}) and Perron-Frobenius theory, the expected population size at time $t$ is asymptotically given by  ${M}(t)\sim \exp(\Omega_0 t)\,\vc{v\,u}^\top,$ as $t \rightarrow \infty$, where $\vc{u}^\top$ and $\vc{v}$ are respectively the left and right Perron-Frobenius eigenvectors of $\exp(\Omega)$; this holds provided that the dominant eigenvalue has multiplicity one, which is generally the case in most practical situations.  
 This implies that for all $i,j$,
\begin{equation} \label{frob}\dfrac{M_{ij}(t)}{({M}(t) \vc{1})_i}\sim\frac{u_j}{\vc{u}^T\vc{1}} \qquad \mbox{as $t \rightarrow \infty$,}
\end{equation}
and we may use \eqref{frob} to obtain the proportion of the population living at each node of the network as time goes to infinity.

\section{Sensitivity analysis}\label{sec:sens}

In practical situations, precise values for the parameters $\{D,\{B_k\}_{k\geq1},\vc d\}$ of the model can be difficult to obtain. It is therefore important to perform a sensitivity analysis of the model with respect to perturbations or errors in the parameters. 
This section addresses the sensitivity analysis of the measures of interest defined in the previous section with respect to the parameters of the model. 

Let $p$ be a parameter of the model (for instance
$ p=T_{ij},\; p=(B_1)_{i;jk},$ or $p=d_i$),
and let $X$ be a measure of the model (for instance $X=M_{ij}(t)$, $X=D_i(t)$, or $X=R_0$).
The \textit{sensitivity} of $X$ with respect to $p$ is defined by the local slope of $X$, considered as a function of $p$:
\begin{equation}\partial_p X=\dfrac{\partial X}{\partial p}.\end{equation} 

The scale of $X$ and $p$ may be different; it is therefore convenient to consider proportional perturbations instead of absolute ones. The proportional response to a proportional perturbation is the \textit{elasticity} (also called \textit{sensitivity index} in the context of mathematical epidemiology  \cite{chitnis2008determining}).
The elasticity
of $X$ with respect to $p$, denoted by $\Upsilon^X_p$, is defined by the ratio of the relative increase of $X$ to the relative increase of $p$:
\begin{equation}\Upsilon^X_p=\dfrac{\partial\log X}{\partial\log p} =\partial_p X\,\dfrac{p}{X}.\end{equation}
The interpretation of the elasticity is as follows: if $\Upsilon^X_p=a$, it follows that if $p$ increases by $r\,\%$, then $X$ increases (or decreases, depending on the sign of $a$) by approximately \begin{equation}\label{elasint}100[\exp(a \log(1+r*0.01))-1]\,\%
\end{equation}(see Appendix~\ref{apel} for more details). Note that for $r=1$, $\log(1.01)\approx 0.01$ so that $100[\exp(a/100)-1]\approx a$ when $|a|$ is small. 

We now derive analytical formulas for the sensitivity of the performance measures defined in the previous section. Here we focus on the matrix case only, as this is the case leading to relevant discussions.

\subsection{Sensitivity of the instantaneous population size}\label{simps}

In order to characterize the sensitivity of the population size generating function $\vc Q(t,\vc s)$, we differentiate \eqref{equdiff} with respect to $p$. This leads to a matrix linear differential equation for $\partial_p \vc Q(t,\vc s)$ which does not have any closed-form solution, see Eq. \eqref{eq28} in Appendix~\ref{apfurther}.

Recall from \eqref{Mtscal} that the mean population size matrix is $M(t)=\exp(\Omega t)$.
Derivatives of the matrix exponential have been investigated by several authors, see for instance \cite{AlMohy},  \cite{VL} and \cite{najfeld}. We present here 
a simple method which provides an exact expression for $\partial_p M(t)$ and involves the computation of a matrix exponential of size $2n\times 2n$.
Let us consider the system of differential equations  \begin{equation}\label{sys}\left\{\begin{array}{rcl}\partial_t M(t)&=&\Omega\,M(t)\\\partial_t\partial_p M(t)&=&\Omega\,\partial_pM(t)+\partial_p\Omega\,M(t),\end{array} \right.\end{equation}
where the first equation is Eq. \eqref{MMscal} satisfied by $M(t)$, and the second equation is obtained by differentiating the first equation with respect to $p$. The system \eqref{sys} may be equivalently rewritten as 
$$\partial_t \left[\begin{array}{c}\partial_pM(t)\\M(t)\end{array}\right]=\left[\begin{array}{cc} \Omega & \partial_p\Omega\\0&\Omega\end{array}\right]\cdot\left[\begin{array}{c}\partial_pM(t)\\M(t)\end{array}\right],$$with initial condition $[\partial_pM(0), M(0)]^\top=[0,I]^\top$, where $0$ and $I$ denote the $n\times n$ zero matrix and identity matrix respectively. This is a new differential equation of the form $\partial_t Y(t)=A\, Y(t)$, of which the solution is given by $Y(t)=\exp(A\,t)\,Y(0)$. Therefore,
\begin{equation}\label{sens1} \partial_pM(t)=[I,0]\cdot\exp\left(\left[\begin{array}{cc} \Omega & \partial_p\Omega\\0&\Omega\end{array}\right]\,t\right)\cdot\left[\begin{array}{c}0\\I\end{array}\right],\end{equation} where \begin{equation}\label{dpT}
\partial_p\Omega = \partial_p T +\sum_{k\geq 1}\, \partial_p B_k\,\sum_{i=0}^k( \vc 1^{(i)}\otimes I \otimes \vc 1^{(k-i)}).\end{equation}

A second method to compute $\partial_pM(t)$, which has the advantage of not requiring any matrix exponential computation, is provided in Appendix~\ref{apfurther}.


\subsection{Sensitivity of the cumulative mean population size}
Focusing now on the sensitivity analysis of the cumulative population size, we observe that, similar to $\partial_p \vc Q(t,s)$, $\partial_p \vc G(t,s)$ satisfies a linear matrix differential equation with non-constant coefficients which does not have any closed-form solution.

The explicit formula for $\partial_p \vc D(t)$ follows easily from the one for $\partial_p M(t)$. Recall the expression for $\vc D(t)$ given in \eqref{Dt} and note that $$\partial_p(-\Omega)^{-1}=(-\Omega)^{-1}\,\partial_p\Omega\,(-\Omega)^{-1}.$$ We therefore obtain
\begin{eqnarray} \label{sens22}\partial_p \vc D(t)& =  &-\partial_pM(t)\,(-\Omega)^{-1}\,\vc d+[I-\exp(\Omega\,t)]\,(-\Omega)^{-1}\,\partial_p\Omega\,(-\Omega)^{-1}\,\vc d+\partial_pM(t)\,\vc 1,\end{eqnarray} where $\partial_pM(t)$ is computed in Section~\ref{simps}.

\subsection{Sensitivity of the extinction probability} Here we derive an explicit expression for $\partial_p\vc q$ in terms of the extinction probability vector $\vc q$. By differentiating \eqref{exteq} with respect to $p$, we obtain
\begin{eqnarray*}
0&= &\partial_p\vc d+\partial_p T\, \vc q+T\,\partial_p \vc q+\sum_k \partial_pB_k\,\vc q^{(k+1)}+ \sum_k B_k\,\sum_{i=0}^k(\vc q^{(i)}\otimes I\otimes \vc q^{(k-i)})\,\partial_p\vc q,
\end{eqnarray*}so that, in the non-critical case,
\begin{equation}
\partial_p\vc q=-\Phi^{-1}\,\Big[\partial_p\vc d+\partial_p T\, \vc q+\sum_k \partial_pB_k\,\vc q^{(k+1)}\Big],
\end{equation}where $$\Phi=T+\sum_k B_k\,\sum_{i=0}^k(\vc q^{(i)}\otimes I\otimes \vc q^{(k-i)}).$$ Note that when $\vc q=\vc 1$,  $\Phi=\Omega$. The non-singularity of the matrix $\Phi$ in the non-critical case is ensured by the next proposition.

\begin{prop}If the Markovian tree is supercritical or subcritical, then the eigenvalues of $\Phi$ all have a strictly negative real part. In the critical case, $\Phi$ is singular.
\end{prop}We omit the proof which follows the same lines as in \cite[Theorem 6]{soph5}.

\subsection{Sensitivity of the extinction criteria}\label{sensR0}

Recall that the eigenvalues of maximal real part of $\Omega$ and $R$, denoted by $\Omega_0$ and $R_0$ respectively, are key quantities to determine whether the branching process almost surely becomes extinct or not. Even though no explicit expression can be written for $\Omega_0$ and $R_0$, analytical expressions can be derived for their sensitivities $\partial_p \Omega_0$ and $\partial_p R_0$, as we show now.

Let $A$ be any generic $n\times n$ matrix (that is, $A$ represents both $\Omega$ and $R$). Let $A_0$ be the eigenvalue of maximal real part of $A$, and let $\vc u^\top$ and $\vc v$ be the left and right eigenvectors of $A$ corresponding to $A_0$, scaled such that $\vc u^\top\vc v=1.$  Then, $$\partial_p A_0= \sum_{i=1}^n\sum_{j=1}^n \dfrac{\partial A_0}{\partial A_{ij}}\dfrac{\partial A_{ij}}{\partial p},$$ where ${\partial A_0}/{\partial A_{ij}}=u_i\,v_j$ \cite[Chapter 9]{caswell1989matrix}, so that $$\partial_p A_0=\vc u^\top\,\partial_p A\,\vc v.$$ Finally, when $A=\Omega$, $\partial_p\Omega$ is given in \eqref{dpT}, and when $A=R$, we use the explicit expression \eqref{exR} for $R$ to obtain an expression for its derivative $\partial_p R$, see \eqref{dpR} in Appendix~\ref{apfurther}.

\subsection{Sensitivity of the asymptotic node frequency}

Recall from Section \eqref{freq} that $\vc u^\top$ and $\vc v$ are the Perron-Frobenius left and right eigenvectors of $\exp(\Omega)$ associated with the dominant eigenvalue $\exp(\Omega_0)$. Let $\{(\lambda_i,\vc u_i^\top, \vc v_i), 1\leq i\leq n\}$ be the full set of eigenvalues of $\exp(\Omega)$, with their corresponding left and right eigenvectors, with $\lambda_1= \exp(\Omega_0)$, $\vc u_1^\top=\vc u^\top$ and $\vc v_1=\vc v$, and assume that the pairs of eigenvectors are normalized such that $\vc u_i^\top \vc v_i=1$ for all $i$. 
Then, the sensitivity of the asymptotic node frequency is given by
\begin{equation}\partial_p\left(\dfrac{\vc u^\top}{\vc u^\top\vc 1}\right)
=\dfrac{\partial_p \vc u^\top(\vc u^\top\vc 1)-\vc u^\top\partial_p\vc u^\top\vc 1}{(\vc u^\top\vc1)^2},
\end{equation}where 
\begin{equation}\label{dpu}\partial_p \vc u^\top= \sum_{i=2}^n \dfrac{\vc u^\top (\partial_p \exp{\Omega})\, \vc v_i }{\exp(\Omega_0)-\lambda_i}\,\vc u_i^\top,\end{equation} which is obtained following
 the same arguments as in \cite[Chapter 9]{caswell1989matrix}, and where $\partial_p \exp{\Omega}$ is computed in Section~\ref{simps}.

\subsection{Illustrative example}
In this section, we illustrate the calculation of a few sensitivity results on a branching process evolving on the simple network with $n=4$ nodes depicted in Figure~\ref{ftoy}. 
\begin{figure}[t!] 
	\begin{center}
	\includegraphics[angle=0,width=3cm]{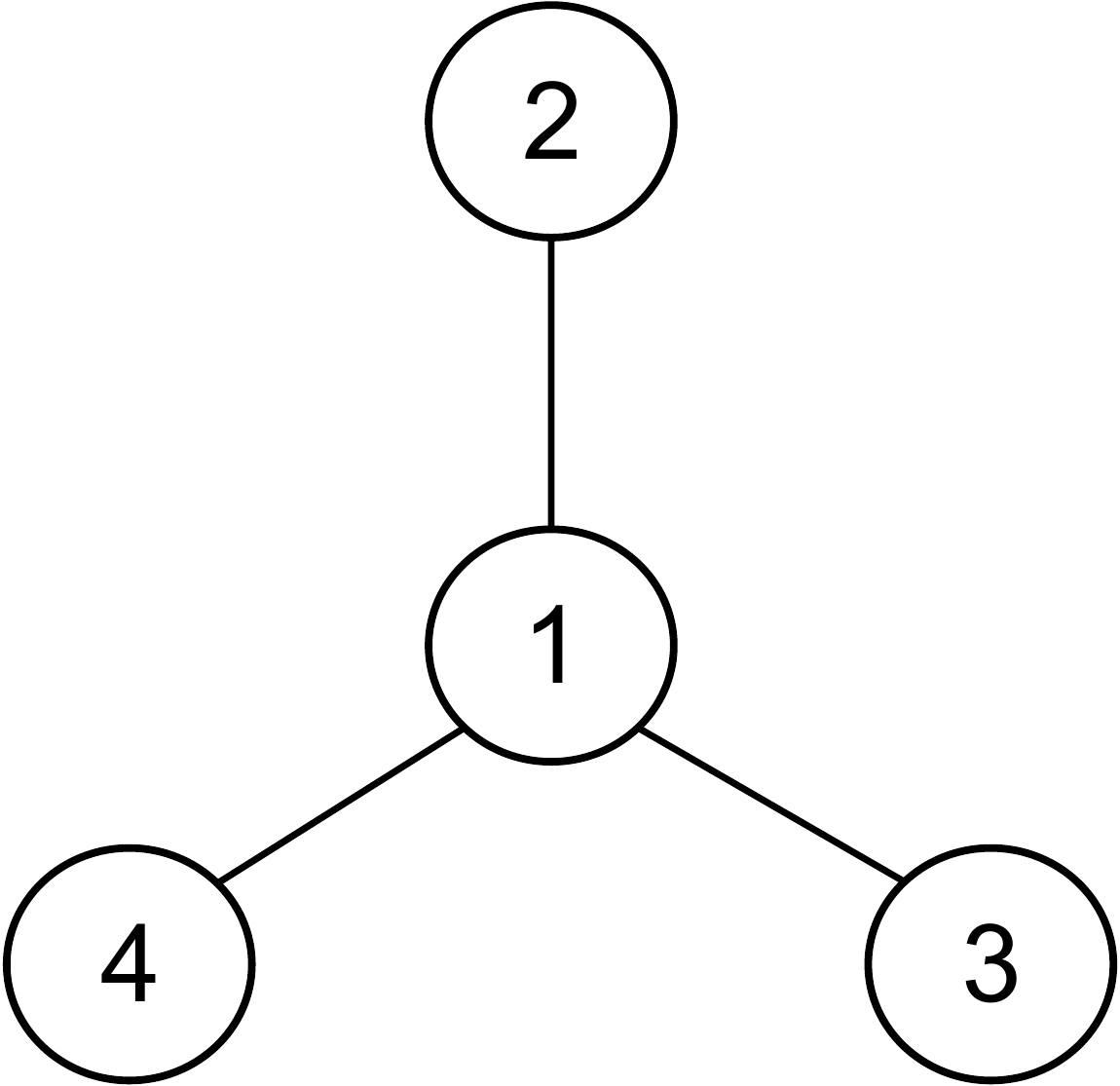} 
	\caption{\label{ftoy}\textbf{Graph of a network with $n=4$ nodes.}}
\end{center}
\end{figure}Assume that every individual generates a single child at each birth event, that is, $B_k=0$ for $k\geq 2$. To simplify the notation, we drop the index 1 in the birth rate matrix $B_1$. The structure of the network induces the following structure for the matrices and vector characterizing the branching process:
$$T=\left[\begin{array}{cccc} T_{11} & T_{12} & T_{13} & T_{14}\\ T_{21 }& T_{22}& 0&0\\T_{31} & 0 & T_{33} & 0\\ T_{41} & 0 & 0&T_{44}\end{array}\right], \quad \vc d=\left[\begin{array}{c} d_1\\d_2\\d_3\\d_4\end{array}\right],$$
$$\arraycolsep=2pt B=\left[\begin{array}{cccc|cccc|cccc|cccc} B_{1;11} & B_{1;12} & B_{1;13} & B_{1;14}& B_{1;21} & B_{1;22} & B_{1;23} & B_{1;24}&B_{1;31} & B_{1;32} & B_{1;33} & B_{1;34}& B_{1;41} & B_{1;42} & B_{1;43} & B_{1;44}\\B_{2;11} & B_{2;12} & 0 & 0& B_{2;21} & B_{2;22} & 0 & 0&0 & 0 & 0 & 0& 0& 0& 0 & 0\\B_{3;11} & 0 & B_{3;13} & 0& 0 & 0 & 0 & 0&B_{3;31} & 0 & B_{3;33} & 0& 0& 0& 0& 0\\B_{4;11} & 0 & 0 & B_{4;14}& 0 & 0 & 0 & 0&0 & 0 & 0 & 0& B_{4;41} & 0 & 0& B_{4;44}\end{array}\right],$$where the diagonal of $T$ is such that $T\vc 1+B\vc 1+\vc d=\vc 0$. 
To further simplify the model, we set $d_i=d$ for all $i$,
$B_{1;i,j}=b_1$ for all $i,j$, $B_{2;i,j}=b_2$ for $i,j=1,2$,  $B_{3;i,j}=b_3$ for $i,j=1,3$, and  $B_{4;i,j}=b_4$ for $i,j=1,4$, for some parameters $d,b_1,b_2,b_3,b_4$. In this special case, the expression \eqref{omdef} for $\Omega$ reduces to $\Omega=T+B(I\otimes \vc 1+\vc 1\otimes I)$ where 
$$T=\left[\begin{array}{cccc} -T_{12} - T_{13} - T_{14}-d-16b_1 & T_{12} & T_{13} & T_{14}\\ T_{21 }& -T_{21 }-d-4b_2 & 0&0\\T_{31} & 0 & -T_{31}-d-4b_3& 0\\ T_{41} & 0 & 0&-T_{41}-d-4b_4\end{array}\right],$$ and
$$B(I\otimes \vc 1+\vc 1\otimes I)=\left[\begin{array}{cccc}8b_1& 8b_1& 8b_1& 8b_1\\
 4b_2& 4b_2&    0&    0\\4b_3&    0& 4b_3&   0\\ 4b_4&    0&   0& 4b_4\end{array}\right].$$ 
 
 We assume that the branching process starts with one individual at node 1 at time $t=0$, and we are interested in the population size at time $t=2$ and in the extinction probability. With the specific parameter values provided in the second column of Table~\ref{tabtoy}, we obtain
 $$M_{11}(2)=308.7914,\quad D_1(2)=4468.1,\quad R_0=1.3426,\quad q_1=0.7390,$$ so the process is supercritical. The elasticities of the above four quantities with respect to each model parameter are shown in Table~\ref{tabtoy}. They can be interpreted using \eqref{elasint}: for instance, if we increase the value of $T_{14}$ by $10\%$, then the value of $M_{11}(2)$ increases by approximately $6.14\%$ and the value of $q_1$ decreases by approximately $0.28\,\%$. We see that the parameters which influence the most the population growth are, in decreasing order: $d$, $b_4$, $b_1$, $b_3$, $b_2$, $T_{14}$, $T_{12}$, $T_{13}$, $T_{21}$, $T_{41}$ and $T_{31}.$ We also see that some transitions between the nodes positively influence the population growth (e.g. from 1 to 4), while other transitions have a negative influence (e.g. from 1 to 2). 
 
\begin{table}[t]
\begin{center} \caption{\label{tabtoy}Sensitivity analysis of the toy example for a given set of parameters.}\medskip
\tabcolsep=0.2cm\begin{tabular}{rrrrrr}
$p$ & Value of $p$ & $\Upsilon_p^{ M_{11}(2)}$& $\Upsilon_p^{ D_{1}(2)}$& $\Upsilon_p^{ R_0}$ & $\Upsilon_p^{ q_1}$
  \\ 
  \hline
 $ T_{12}$ &2 &-0.4683&   -0.3764&   -0.0175 &   0.0198\\
  $T_{13}$ &6 & -0.4679  & -0.3126 &  -0.0133  &  0.0151\\
  $T_{14}$ &8 &0.6255  &  0.6723  &  0.0307  & -0.0298\\
  $T_{21}$ &1 & 0.1564 &   0.1247 &   0.0065  & -0.0062\\
 $ T_{31}$ &1 &0.0733  &  0.0484 &   0.0023&   -0.0022\\
  $T_{41}$ &1 & -0.0826 &  -0.0891 &  -0.0044  &  0.0034\\
  $d $&10 &-20.0000 & -17.9012 &  -0.9804  &  1.0131\\
  $b_1$ &1 &7.4966  &  6.9555  &  0.2608 &  -0.3174\\
  $b_2$ &2 & 3.1233  &  2.7226 &   0.1227 &  -0.1286\\
  $b_3$ &3 & 6.5588  &  5.7174  &  0.2448 &  -0.2451\\
  $b_4$ &4 &  9.9944   & 8.7123 &   0.3478  & -0.3220
  \end{tabular} 
\end{center}
\end{table}

\section{Application: Sensitivity analysis of an influenza-like epidemic spreading on a network of cities}\label{app}

In this section we apply the sensitivity results developed in the previous section to a stochastic model for the initial spread of an influenza-like epidemic among a set of cities connected by air routes in the United States. This allows us to determine the relative importance of the model parameters to the disease spread. For this purpose, the early stages of the epidemic are approximated by a branching process evolving on a network as described in Section~\ref{bpn}.

Appropriate branching processes are frequently used to approximate the process of infection during the early stages of a general stochastic epidemic in a large closed and homogeneously mixing population. 
In theory, during the course of a major epidemic, the epidemic grows like a branching process until about $\sqrt{N}$ members of the population of $N$ individuals become infected, see \cite{bd95}. The accuracy of the approximation actually depends on our own error criteria and is discussed in Appendix~\ref{valapprox}. As shown in Section~\ref{sec:sens}, the tractability of branching processes make them perfectly suitable for a sensitivity analysis.

\subsection{Description of the model}\label{sedes}
We assume a constant, large, and homogeneously mixing population in each city. We use a branching process on a network, characterized by a set of matrices $\{T,\{B_k\}_{k\geq1},\vc d\}$ that we detail below, to model the spread of the disease among the cities during the early stages of the epidemic. As discussed in Appendix~\ref{valapprox}, our model of branching process is suitable as long as the number of susceptible individuals is much larger than the number of infected individuals in each city, which justifies why we focus on the early stages of the disease. In the rest of this section, the time unit is the day. 

Individuals in the branching process correspond to people infected by the disease, and each of the $n$ nodes of the network corresponds to a city in the U.S. An individual makes a transition from node $i$ to node $j$ in the network when he travels by plane from city $i$ to city $j$.
Cities are connected through a symmetric $n\times n$ air travel matrix $A$ where the entry $A_{ij}$ is the average number of passengers per day from city $i$ to city $j$ for $i\neq j$ (by convention we set $A_{ii}=0$). The symmetry assumption for $A$ implies that the population of each city remains constant, see for instance \cite{chitnis2008determining} and \cite{gekg04}.
The travel rate per individual per day from city $i$ to city $j$ is obtained by dividing the daily average number of passengers going from city $i$ to city $j$ by the population of city $i$, that we denote by $N_i$; the nondiagonal entries of the transition rate matrix $T$ are therefore given by
\begin{eqnarray}\label{tr}T_{ij}&=&\frac{A_{ij}}{N_i},\qquad i\neq j.\end{eqnarray} For the purpose of our application, we call the matrix $T$ the \textit{travel rate matrix}. 

In a first approach, we assume that infected people transmit the disease to new individuals within cities only, not during their plane travel. An individual at node $i$ infects a new individual at rate $\beta_i$; this parameter thus corresponds to the average number of infectious contacts per day by an infected individual in city $i$. The rates $\beta_i$ depend on the cities and are gathered in a \textit{transmission rate vector} $\vc \beta$. Since only one transmission may occur at a time, the sequence $\{B_k\}_{k\geq1}$ of birth rate matrices in the branching process is such that the only nonzero entries of the matrix $B_1$ are $$(B_1)_{i;ii}=\beta_i \quad\mbox{ for }i=1,\ldots,n,$$ and $$B_k=0\quad\mbox{ for }k\geq 2.$$  

Finally, infected individuals may be ``removed'' from the population by recovering (or dying) at rate $d_i$ in city $i$, which corresponds to the inverse of the mean time of contagion of an infected individual in city $i$. These rates correspond to the death rates in the branching process and form the \textit{removal rate vector} $\vc d$.
Recall from Section~\ref{bpn} that the diagonal elements of the matrix $T$ are then such that $T\vc 1+B_1\vc 1+\vc d=\vc 0$, that is, $T_{ii}=-\beta_i-d_i-\sum_{j\neq i} T_{ij}$.

Due to the simple structure of the matrix $B_1$, we have $B_1(\vc 1\otimes I)=B_1(I\otimes \vc1)=\diag(\vc\beta)$, so that the expressions for the matrices $\Omega$ and $R$, given in \eqref{omdef} and \eqref{exR} respectively, simplify to
$$\Omega=T+2\,\diag(\vc\beta),\quad R=-\left[I+T^{-1}\diag(\vc\beta)\right]^{-1}\,T^{-1}\diag(\vc\beta).$$ In the epidemic context, $R$ is the matrix of mean number of secondary infections generated in each city by an average infected individual in an entirely susceptible population. The dominant eigenvalue $R_0$ of $R$ is called the \textit{basic reproduction number}, which is a key quantity for determining whether an infection can invade and persist with a positive probability: the infection can invade the population if and only if $R_0>1$ (supercritical case). When $R_0<1$ (subcritical case) the disease simply dies out, and when $R_0=1$ (critical case) the disease becomes endemic, meaning that the proportion of infected individuals remains constant over time.

In a second approach, we take into account possible transmission on board airplanes. The description of the model in this case is given in Appendix~\ref{obtd}.

\subsection{Data}\label{data}

Air travel data on the daily average number of passengers for each city-pair $(i,j)$ among a group of $n=114$ cities are obtained from the Domestic Airline Fares Consumer Report of the US Department of Transportation for the first Quarter of 2011 \cite{DOT}. This report provides information on the 1,000 largest city-pair markets in the 48 
contiguous states, among which the number of 
one-way passenger trips per day.  These markets involve 114 cities and
account for about 75 percent of all 48-state passengers and 70 percent of total domestic 
passengers.
Data on the metropolitan population of the American cities are taken from the United States Census Bureau (2011 estimates). The travel rates are computed according to \eqref{tr}. 

%

Disease parameters within cities are taken from \cite{gekg04}. In that paper, the contact rate between susceptibles and infectives, $\beta^*$, is generally estimated to be 1.0. Like many respiratory diseases, influenza exhibits a seasonal pattern with a low summer and a high winter incidence. Here, we chose to only focus on the autumn-winter period going from October to March. Cities are divided into five general zones based on 
the average number of heating degree-days and cooling degree-days, see for instance \cite{climate}. We applied a scaling factor to $\beta^*$ which depends on the climate zone of each city and ranges between 0.85 and 1.1 in the autumn-winter period. This provides a specific value $\beta_i$ for each city $i$. The complete list of cities with their corresponding metropolitan population and transmission rate is provided in Table~\ref{fig:data} in Appendix~\ref{apcd}.

Finally, the average length of the infectious period is estimated to be 2.95 days independently of the city \cite{gekg04}, which leads to $ d_i=d=1/2.95$ for all $i$.

\subsection{Results and Discussion}\label{illu}

In this section, we use the results of Section~\ref{sec:sens} to compute the elasticity of the mean cumulative epidemic size $\vc D(t)$ and of the basic reproduction number $R_0$, with respect to the parameters of the model, over the first two weeks of the epidemic, that is, for $t\in[0,14]$ (as discussed in Appendix~\ref{valapprox}, the branching process approximation is reasonable during this time period).

We present our results for a small number of cities, namely New York, Orlando, Chicago, and San Francisco. As shown in Table~\ref{tab2b}, this set of cities covers all possible values of transmission rates. Table~\ref{tab2b} also presents the travel rates from New York to the three other cities, and the mean cumulative size of the epidemic after 14 days if it started in each of the four cities. We obtain $\Omega_0=2.1401$ and $R_0=3.2433$, meaning that the epidemic breaks out with positive probability. Note that these values are particularly high due to the relatively high transmission rates during the autumn-winter period (between 0.85 and 1.1).

\begin{table}[t]
\begin{center} \caption{\label{tab2b} Travel rates from city $i$ (which we assume here to be New York) to city $j$ ($T_{ij}$, travel frequency per individual per day from $i$ to $j$), transmission rates in city $j$ ($\beta_j$, average number of new infection per individual per day in city $j$), and mean cumulative epidemic size after 14 days given that the disease started in city $j$ ($D_j(14)$),
for four selected cities $j$.}
\medskip
\begin{tabular}{lcccc}
  City $j$ & $T_{ij}$ & $\beta_j$ & $D_j(14)$
  \\ 
  \hline
  New York & -- &  $1.1$ & $ 5.93\cdot10^4$\\
  Chicago & $3.46\cdot10^{-4} $ & $1.1$ & 	$5.90\cdot10^4$\\
  San Francisco & $3.20\cdot10^{-4} $  & $1.02$  &$2.13\cdot10^4$\\
  Orlando & $3.64\cdot10^{-4} $ & $0.85$ &$6.31\cdot10^3$
\end{tabular} 
\end{center}
\end{table}

\paragraph*{Elasticity of the epidemic size with respect to the transmission rates.} 
Figure~\ref{fielbet3} shows the elasticity of the mean cumulative size of the epidemic with respect to the transmission rate in the city where the disease started, for three cities with different transmission rates. We see that the elasticities are non-negligible and the curves appear in the same order as the transmission rates. Using formula \eqref{elasint}, we see that an increase by $1\%$ of the transmission rate in New York will induce an increase by approximately $14\%$ of the mean cumulative epidemic size after two weeks, while an increase by $1\%$ of the transmission rate in Orlando will induce a corresponding approximate increase of $5\%$ only.
We observe that the elasticities increase over time except, interestingly, for Orlando where the elasticity reaches a maximum and then starts decreasing. This can be interpreted as an effect of the relatively low transmission rate associated with Orlando and the high travel rate from that city: if the epidemic starts in Orlando, after about ten days the growth of the disease in other cities connected to Orlando starts exceeding the infected population in Orlando (due to higher transmission rates in these cities, see also Figure~\ref{fi2} in Appendix~\ref{valapprox}). There is thus a time from which the growth of the disease depends less on the transmission rate in Orlando than on the transmission rate in other cities, and therefore a small increase in the transmission rate of Orlando would start having a decreasing impact on the size of the epidemic.

\begin{figure}[t]\begin{center}
 
 \begin{minipage}{0.4\textwidth}
    \begin{figure}[H]
        \hspace*{-1.6cm}  \includegraphics[scale=0.21]{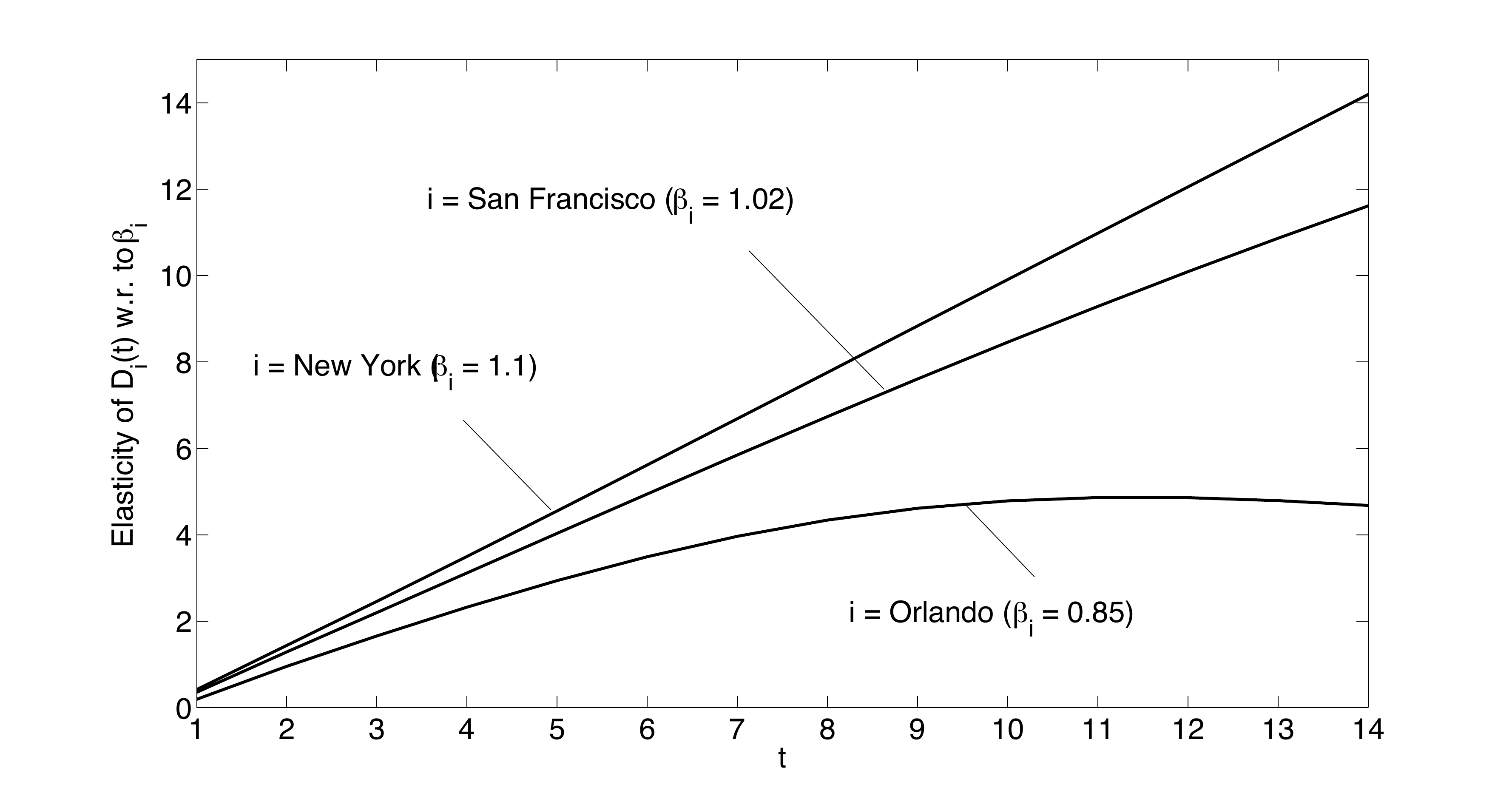}
        \vspace*{-0.5cm}\caption{\label{fielbet3} \textbf{Elasticity of the mean cumulative epidemic size with respect to the transmission rates.} Elasticities of the mean cumulative epidemic size at time $t$ if the disease starts in city $i$ ($D_{i}(t)$) with respect to the transmission rate in city $i$ ($\beta_i$) (left) for a sample of origin cities $i$.
	        \vspace{0.2cm}}
  
    \end{figure}
\end{minipage}
\hspace{3ex} 
\begin{minipage}{0.4\textwidth}
    \begin{figure}[H]
      \hspace*{-1.3cm}  \includegraphics[scale=0.21]{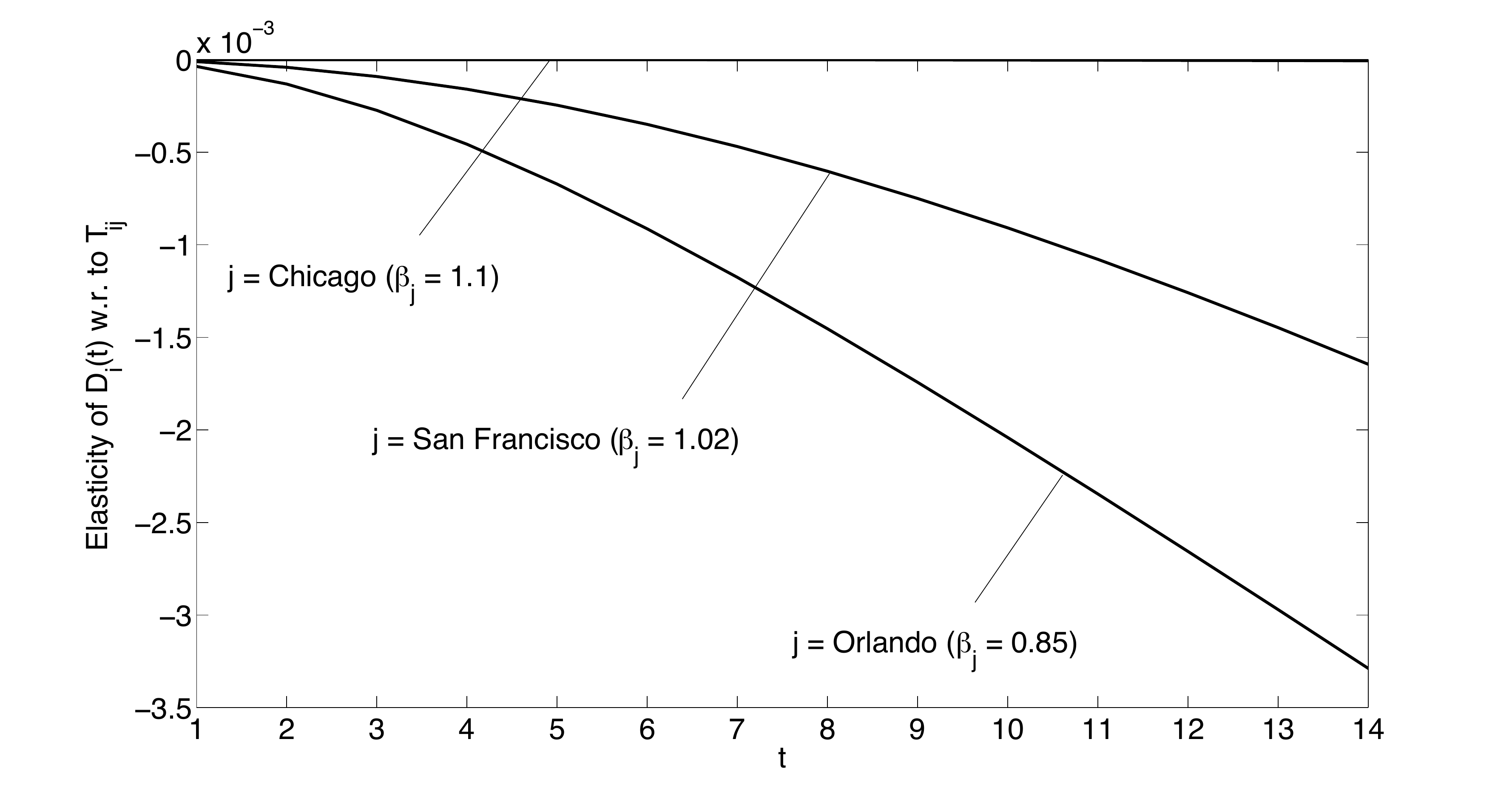}
        \vspace*{-0.5cm}\caption{\label{fielD01bis}\textbf{Elasticity of the mean cumulative epidemic size with respect to the travel rates.} Elasticities of the mean cumulative epidemic size at time $t$ if the disease starts in city $i$ ($D_{i}(t)$) with respect to the travel rate from city $i$ to city $j$ ($T_{ij}$), with New York as the origin city $i$ of the disease, and for a sample of cities $j$. }
    \end{figure}
\end{minipage}\vspace*{-0.5cm}
\end{center}
\end{figure}

\paragraph*{Elasticity of the epidemic size with respect to the travel rates.}
Figure~\ref{fielD01bis} shows the elasticity of the mean cumulative number of infected individuals with respect to three travel rates from New York (where the epidemic is initiated): the travel rates to Chicago, to San Francisco, and to Orlando (see Table~\ref{tab2b} for their values). We see that the elasticities are significantly smaller in absolute value than when they were computed with respect to the transmission rates. A small perturbation in the travel rates has therefore less impact on the dynamics of the disease than an equivalent perturbation in the transmission rates. 
We also observe that the elasticities are all negative, which suggests that increasing the travel rates out of New York, the origin city of the disease, may be beneficial from a sanitary point of view. This result might be surprising but it makes sense because the transmission rate in New York is the highest ($\beta_i=1.1$), and we assumed that there is no transmission during travel (we refer to Appendix~\ref{obtr} to see what happens when on-board transmission is taken into account). 
 Indeed, increasing the travel rates from New York to a city with the same transmission rate, such as Chicago, will have almost no effect on the dynamics of the disease, while increasing the travel rates from New York to a city with a very low transmission rate, such as Orlando, will have a higher impact on the size of the epidemic. This argument further suggests that a smaller transmission rate in the destination city induces a larger elasticity (in absolute value), as confirmed in the figure. 

\paragraph*{Elasticity of $R_0$.}
Recall that the basic reproduction number $R_0$ is an important threshold scalar quantity which measures the initial disease transmission on the whole network. In order to study the sensitivity of $R_0$ with respect to the parameters of the model, we use the results of Section~\ref{sensR0}, where the expression for $\partial_p R$ (provided in \eqref{dpR} in Appendix~\ref{apfurther}) simplifies to
\begin{eqnarray}\label{ro1ap}\partial_p R\hspace{-0.2cm}&=&\hspace{-0.2cm}\left[I+T^{-1}\diag(\vc\beta)\right]^{-1}\left[T^{-1}(\partial_p T) T^{-1}\diag(\vc\beta)-T^{-1}\partial_p\diag(\vc\beta)\right]\left\{\left[I+T^{-1}\diag(\vc\beta)\right]^{-1}+I\right\}.\qquad\end{eqnarray}

 Not surprisingly, the local sensitivity of $R_0$ with respect to the transmission rate in one single city, $\beta_i$, or with respect to the travel rate between two specific cities, $T_{ij}$, is almost negligible (of the order of $10^{-6}$ or less).
\begin{table}[t]
\begin{center}
\caption{\label{tabR0}Elasticity of the basic reproduction number $R_0$ with respect to some model parameters which are constant for all cities. }
\medskip
\begin{tabular}{rccc}
$p:$ &  $\beta^*$&$d$& $c$ 
   \\\hline
  $\Upsilon^{R_0}_p:$ &   $1.2357$ & $-1.2350$ & $-6.62\cdot 10^{-4}$ \\
  \end{tabular} 

\end{center}
\end{table} 
We also performed a global sensitivity analysis of $R_0$ with respect to some parameters which are assumed to be the same for all cities, namely the contact rate $\beta^*=1$ and the removal rate $d=1/2.95$. The results are shown in Table~\ref{tabR0}.
We see that a $1\%$ increase in the contact rate $\beta^*$ (or, equivalently, in the transmission rate of \textit{all} cities together) results in a $1.24\%$ increase in $R_0$.
We also note that, in absolute value, the elasticities of $R_0$ with respect to $\beta^*$ and with respect to $d$ are very close to each other, $R_0$ being slightly more sensitive with respect to $\beta^*$ than with respect to $d$.

In order to analyze the relative change in $R_0$ when \textit{all} travel rates increase by the same factor, we write $T_{ij}=c\cdot T_{ij}$ ($i\neq j$) where the constant $c$ is equal to 1, and we compute the sensitivity of $R_0$ with respect to $c$. The result is given in the last column of Table~\ref{tabR0} and shows that if all travel rates increase by $1\%$, then $R_0$ decreases by $6.62\cdot 10^{-4}\%$. This suggests that increasing the frequency of all domestic travel (by the same percentage) would globally  slow down the epidemic very slightly. This is again a surprising result which comes from the fact that each city has its own transmission rate, and increasing the frequency of all domestic travel (without on-board transmission) would, among other things, increase the travel rates out of the cities with a high transmission rate.
Finally, we see again that the sensitivity is significantly lower with respect to the travel rates than with respect to the transmission rates.

\begin{rem}Note that the sensitivity results are very dependent on the disease under consideration: since the climate plays a fundamental role in the spread of influenza, the growth of the disease is sensitive to the travel rates because different cities in different climate zones exhibit different transmission rates. However, other types of infectious diseases may spread independently of the climate, in which case the air travel network would have much less impact on the epidemic growth, if we assume that there is no in-flight transmission. \end{rem}

\medskip
A discussion of the results in the case where on-board transmission is taken into account is given in Appendix~\ref{obtr}.

\subsection{Vaccination}

Vaccination can be used to control the spread of the disease within a population and hence eradicate it.  From a public authority's perspective, one may want to solve two questions:\begin{itemize} \item[(i)] what is the smallest proportion of individuals that need to be vaccinated in order to prevent an epidemic outbreak, and \item[(ii)]  once this number is reached, how many new infections do we prevent with each additional vaccinated individual? \end{itemize}In this section, we will tackle those two issues by considering the simple problem where there is no on-board transmission and vaccination is uniform over the population. Appendices~\ref{apvaccd} and~\ref{apvacobt} respectively deal with the cases where vaccination is treated differently in each city and where there is on-board transmission. 

As we do not consider age patterns, we shall assume that each individual is vaccinated with some probability $r$, independently of his age and independently of the city. The parameter $r$ thus also denotes the fraction of the total American population which is vaccinated. Vaccination reduces susceptibility, which, in our branching process approximation model, is equivalent to considering
the new transmission rate vector $\vc\beta_v=(1-r)\,\vc\beta$ (in the sequel, we write the subscript $v$ each time a quantity depends on this new vector).
The minimal fraction of individuals that should be vaccinated is solution of the minimization problem
\begin{equation*}
\begin{aligned}
& \underset{r}{\text{minimize}}
& & r\sum\limits_{i}{N_i}\\
& \text{subject to}
&& \lambda_{max}(\Omega_v)\leq 0,
\end{aligned}
\end{equation*} where $\lambda_{max}(\Omega_v)$ denotes the Perron-Frobenius (dominant) eigenvalue of the matrix $\Omega_v$.
The solution of this problem can be written out 
explicitly and the critical value obtained for our model is $r_c= 0.69$, which means that at least $69\%$ of the population has to be vaccinated to avoid an outbreak (see Appendix~\ref{apvacu} for more details). Assuming that the total population in the US is $3.12\cdot 10^8$, the critical number of individuals to vaccinate is $3.12\cdot 10^8\, r_c\approx2.16\cdot 10^8$; this solves Question (i). Given that at least that fraction of the population is vaccinated, the answer to Question (ii) is then given by the sensitivity of the mean total cumulative number of infected individuals until eradication with respect to $r$.

Let $D_i$ represent the mean total cumulative size of the infection given that it was initiated by a first individual in city $i$, and $\vc D=(D_i)=\lim_{t\rightarrow\infty} \vc D(t)$, where $\vc D(t)$ is defined in \eqref{Dt}. Since the dominant eigenvalue of $\Omega_v$ is negative, $\lim_{t\rightarrow\infty} \exp(\Omega_v t)=0$, and from \eqref{Dt} we obtain
\begin{equation} \vc D=(-\Omega_v)^{-1}\,\vc d,\end{equation} where $\vc D$ depends on $r$ through $\Omega_v$. Therefore, the sensitivity of $\vc D$ with respect to the fraction $r$ of vaccinated people is given by
\begin{eqnarray*}\frac{\partial \vc D}{\partial r} &=&-(-\Omega_v)^{-1}\,\partial_r\Omega_v\,(-\Omega_v)^{-1}\,\vc d\\&=&
-(-\Omega_v)^{-1}\,\diag(\vc\beta)\,(-\Omega_v)^{-1}\,\vc d.\end{eqnarray*}
From this formula, we can approximate the mean number of prevented infections per additional vaccine given that a proportion $r> r_c$ of the population is vaccinated: one more vaccine leads to $\partial r=($total population$)^{-1}=(1/312)\,10^{-6}$, so that $$\partial \vc D=-(1/3.12)\,10^{-8}\,(-\Omega_v)^{-1}\,\diag(\vc\beta)\,(-\Omega_v)^{-1}\, \vc d.$$

Assume that we are in an \textit{almost critical} regime, that is, the initial fraction $r$ of vaccinated individuals is slightly larger than $r_c$; consider, for instance, a value $r$ such that $3.12\cdot 10^8\, (r-r_c)=10$, that is, such that the critical vaccinated population is increased by 10. The benefit of one additional vaccine in this regime is shown in Table~\ref{ta8}, in which we compare the mean total cumulative size of the infection in the almost critical regime given the origin city $i$ of the disease, $D_i$, with the approximate number of prevented infections if we introduce one additional vaccine in the population, $\partial D_i$, for three origin cities of the disease.
\begin{table}[t]
\begin{center}\caption{\label{ta8}Uniform vaccination case in the almost critical regime where $r$ is such that $3.12 \cdot 10^8\,(r-r_c)=10$: the mean total cumulative size of the infection given the origin city of the disease ($D_i$), and the number of infections prevented by introducing one additional vaccine ($\partial D_i$), for three origin cities of the disease}
\medskip
\begin{tabular}{lcc}
Origin city $i$ & $D_i$& $\partial D_i$\\ 
\hline New York& $1.1\cdot 10^4$ & $8.8\cdot 10^2$  \\ 
 San Francisco& $2.6\cdot 10^3$& $2.0\cdot 10^2$   \\ 
Orlando& $3.3\cdot 10^4$ & $2.6\cdot 10^3$
\end{tabular}\end{center}
\end{table}
We see that the introduction of one additional vaccine in this almost critical regime would reduce the size of the epidemic by a bit less than $10\%$. The number of prevented infections per additional vaccine is thus very large when the initial number of vaccinated people is close to the critical number; in other words, in the almost critical regime, the growth of the disease is highly sensitive with respect to the vaccination ratio. The values of $\partial D_i$ decrease very rapidly when the initial fraction of vaccinated people $r$ becomes larger than the critical value $r_c$. This is shown in Figure~\ref{fivac1} for the three cities considered in Table~\ref{ta8}, where we depict the number of people who would escape from the disease if we introduce one additional vaccine in the population as a function of the difference $3.12 \cdot 10^8\,(r-r_c)$ between the initial number of vaccinated people and the critical number. 
\begin{figure}[t!] 
	\begin{center}
	\includegraphics[angle=0,width=8cm]{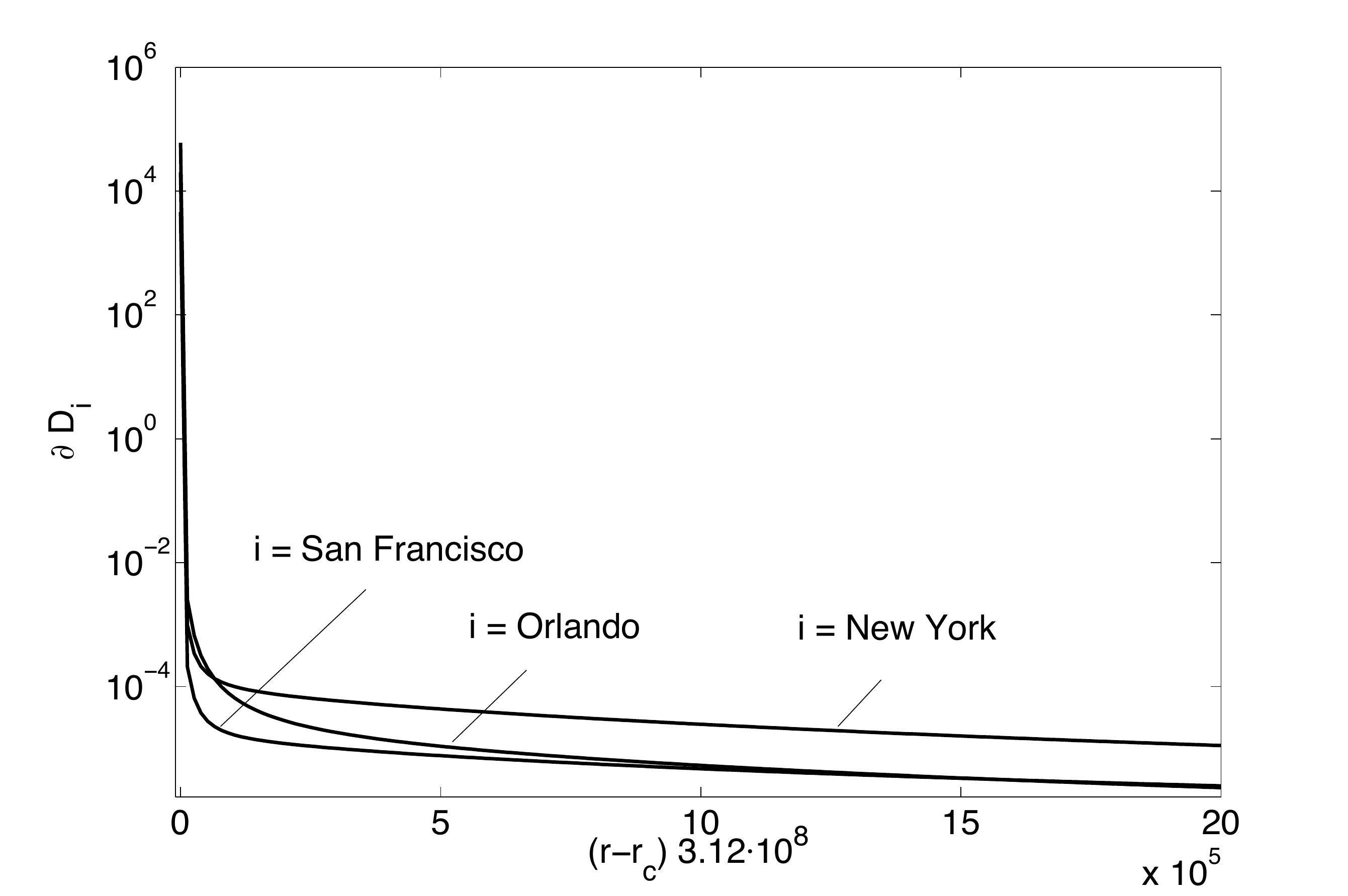} 
	\caption{\label{fivac1}\textbf{Uniform vaccination. }Number of prevented infections per additional vaccine as a function of the difference between the initial number of vaccinated people ($3.12\cdot 10^ 8\, r$) and the critical vaccinated population ($3.12\cdot 10^ 8\, r_c$), for three origin cities of the disease.}
	
\end{center}
\end{figure}

\section{Conclusion} Explicit formulas are derived for the sensitivity analysis of a model of Markovian branching process evolving on a fixed network. This tractable stochastic process is then used to model the early stages of a seasonal influenza-like disease speading on a network of cities in the United States. This approach allows us to study the sensitivity of the size of the epidemic and of the basic reproduction number with respect to the transmission and domestic travel rates in an analytic way, that is, without the use of simulations. Our analysis highlights the differences in the sensitivities with respect to the different parameters, confirming that a precise estimation of the transmission rates is far more important than a precise estimation of domestic mobility data.
We also treat an extension of the epidemic model by considering vaccination campaigns. In this case, a sensitivity analysis enables us to calculate the marginal gain of one additional vaccine, showing how, at the almost critical regime, each additional vaccination prevents a large number of possible infections. 

Note that in the epidemic application, we chose to focus mainly on the sensitivity of expected quantities, such as the mean epidemic size, which can also be described by a linearised SIR model, for simplicity and ease of comparison with those popular models. However our methodology also applies to more sophisticated quantities not deducible from a deterministic model, such as the variance of the epidemic size, or the probability that the epidemic dies out before a given time. We have considered the example of influenza,  which is suitable for our purpose because people usually continue to travel when they are infectious, however our model and methods are not restricted to this particular disease.

Our methodology can be applied to a variety of other Markovian models involving reproduction, death and (discrete) node transition events. Note that the nodes of the network may not only represent geographic locations but also more abstract features such as physiological states of individuals, see for example \cite{hautphenne2012markovian}. A possible extension of the model would be to generalize the results of Sections~\ref{bpn} and~\ref{sec:sens} to a dynamic network (instead of a fixed one); this is the topic of ongoing research.

 \section*{Funding} This work was supported by the Belgian Programme of Interuniversity Attraction Poles, initiated by the
Belgian Federal Science Policy Office; an
Action de Recherche Concert\'ee (ARC) of the French Community of Belgium; the Fonds National de la Recherche Scientifique (F.R.S. - F.N.R.S.); and the Australian Research Council, [grant Nb. DP110101663 to S.H.]. 

\section*{Acknowledgments} 

We thank the referees for their constructive comments which have improved the presentation of the paper, as well as Prof. Peter Taylor and Dr. Joshua Ross who have carefully read and commented the paper.

\appendix
\numberwithin{equation}{section}
\numberwithin{figure}{section}

\section{Further results on the branching process on a network and its sensitivity}\label{apfurther} In this section we gather some technical results related to Sections~\ref{bpn} and~\ref{sec:sens}.

\subsection{The matrix $R$}

Recall from Section~\ref{extcrit} that the entry $R_{ij}$ of the matrix $R$ is the expected total number of children born at node $j$ from a parent born at node $i$, during the entire lifetime of the parent.
For each $k\geq 1$, define the $n\times n^{k+1}$ matrix $$\Psi_k=-T^{-1}B_k,$$ whose entry $(\Psi_k)_{i;j_1\,j_2\ldots\,j_k\,j}$ records the probability that an individual at node $i$ eventually gives birth before dying (potentially after moving to other nodes), and the first reproduction event involves the birth of $k$ children and the parent moving to node $j$, while the children start
their life at nodes $j_1$, $j_2,\ldots,j_k$ respectively. Note that nodes $j$, $j_1$, $j_2,\ldots,j_k$ are not necessarily adjacent to node $i$ here.
An explicit expression for $R$ can be given in terms of the matrices $\Psi_k$, as shown in the next proposition.

\begin{prop}\label{matR} The matrix $R$ of expected total progeny size is given by
\begin{equation}\label{exR}R=\Big[I-\sum_{k\geq 1} \Psi_k(\vc 1^{(k)}\otimes I)\Big]^{-1}\,\sum_{k\geq1} \Psi_k \sum_{i=0}^{k-1} (\vc 1^{(i)}\otimes I\otimes\vc 1^{(k-i)}).\end{equation}
\end{prop} 

\begin{proof}
We develop the entry $R_{ij}$ by conditioning on the first reproduction event happening to the parent at node $i$: 
\begin{eqnarray*} R_{ij}&=& \sum_{k\geq 1} \sum_{j_1}\ldots \sum_{j_k}\sum_{\ell} (\Psi_k)_{i;j_1\,j_2\ldots\,j_k\,\ell}(\delta_{j_1,j}+\ldots+\delta_{j_k,j}+R_{\ell j}),\end{eqnarray*}where $\delta_{i,j}=1$ if and only if $i=j$. In matrix form, this becomes
$$R= \sum_{k\geq 1}\Psi_k\left[ \sum_{i=0}^{k-1} (\vc 1^{(i)}\otimes I\otimes\vc 1^{(k-i)})+(\vc 1^{(k)}\otimes I)R\right],$$ and since the matrix $\sum_{k\geq 1} \Psi_k(\vc 1^{(k)}\otimes I)$ is substochastic, we can write
$R$ explicitly as in \eqref{exR}. \end{proof} 
\medskip

Using the explicit expression for $R$ given in \eqref{exR}, the sensitivity of $R$ with respect to a parameter $p$ is given by
\begin{eqnarray}\nonumber
\partial_p R&=& \Big[I-\sum_{k\geq 1} \Psi_k(\vc 1^{(k)}\otimes I)\Big]^{-1} \Big[\sum_{k\geq 1} \partial_p\Psi_k(\vc 1^{(k)}\otimes I)\Big]\Big[I-\sum_{k\geq 1} \Psi_k(\vc 1^{(k)}\otimes I)\Big]^{-1}\,\sum_{k\geq1} \Psi_k \sum_{i=0}^{k-1} (\vc 1^{(i)}\otimes I\otimes\vc 1^{(k-i)})\\\label{dpR}&&+ \Big[I-\sum_{k\geq 1} \Psi_k(\vc 1^{(k)}\otimes I)\Big]^{-1}\,\sum_{k\geq1} \partial_p\Psi_k \sum_{i=0}^{k-1} (\vc 1^{(i)}\otimes I\otimes\vc 1^{(k-i)}),
\end{eqnarray}where $
\partial_p\Psi_k=\partial_p[-T^{-1}B_k]= -T^{-1}[\partial_p T\;\Psi_k+\partial_p B_k]$.

\subsection{The sensitivity of $\vc Q(t,\vc s)$ and $\vc G(t,s)$}

Here we are interested in characterizing $\partial_p \vc Q(t,\vc s)$ and $\partial_p \vc G(t, s)$.
By differentiating \eqref{equdiff} with respect to $p$, we obtain
\begin{eqnarray}\nonumber\partial_t\partial_p \vc Q(t,\vc s)&=&\partial_p\vc d+\partial_p T \vc Q(t,\vc s)+T\partial_p \vc Q(t,\vc s)+\sum_k\partial_p B_k \vc Q(t,\vc s)^{(k+1)}\\\nonumber&&+\sum_k B_k\sum_{i=0}^k(\vc Q(t,\vc s)^{(i)}\otimes I\otimes \vc Q(t,\vc s)^{(k-i)})\partial_p \vc Q(t,\vc s)\\\label{eq28}&=&A(t,\vc s)+B(t,\vc s)\partial_p \vc Q(t,\vc s)\end{eqnarray} where \begin{eqnarray*}A(t,\vc s)&=&\partial_p\vc d+\partial_p T \vc Q(t,\vc s)+\sum_k\partial_p B_k \vc Q(t,\vc s)^{(k+1)}\\B(t,\vc s)&=&T+\sum_k B_k\sum_{i=0}^k(\vc Q(t,\vc s)^{(i)}\otimes I\otimes \vc Q(t,\vc s)^{(k-i)}).\end{eqnarray*} Unfortunately
the matrix linear differential equation \eqref{eq28} for $\partial_p \vc Q(t,\vc s)$ does not have any closed-form solution because the matrices $B(t,\vc s)$ do not necessarily commute for different values of $t$. 

Using the same argument and Eq. \eqref{eqdiff2}, we show that $\partial_p \vc G(t,s)$ satisfies the same matrix linear differential equation as $\partial_p \vc Q(t,\vc s)$, the only difference being that the term $\partial_p\vc d$ is replaced by $\partial_p\vc d s$.

\subsection{An approximation for the sensitivity of $M(t)$}
An alternative method to compute $\partial_pM(t)$ provides a new approximation for $\partial_pM(t)$ which relies on the integral representation of the derivative of the matrix exponential \cite{najfeld}: \begin{equation}\label{repr}\partial_pM(t)=\partial_p \exp(\Omega t)=\int_0^t \exp(\Omega \tau) \,\partial_p\Omega\,\exp(\Omega (t-\tau))\,d\tau.\end{equation} 

\begin{prop} The matrix $\partial_pM(t)$ can be approximated as \begin{equation}\label{met2}{\partial_p M(t) }\approx\frac{1}{c}\,\exp(\Omega_0 t)\,\Delta_{v}\sum_{i=0}^{K(t)}\sum_{k=i+1}^{K(t)+1}\exp(-ct)\,\frac{(ct)^k}{k!}\,\hat{\Theta}^i\,\Delta_{v}^{-1}\,\partial_p\Omega\,\Delta_{v}\,\hat{\Theta}^{k-i-1}\,\Delta_{v}^{-1}\end{equation} where $c$, $\Delta_{v}$, $K(t)$ and $\hat{\Theta}$ are defined in the proof.
\end{prop}

\begin{proof} The integral in \eqref{repr} is approximated by using a duality argument (as in \cite{hln1}) and the uniformization technique for continuous-time Markov chains (see for instance \cite{ross}), as we detail now.

Let $\vc u^\top$ and $\vc v$ be the left and right eigenvectors of $\Omega$ corresponding to the Perron-Frobenius eigenvalue $\Omega_0$, normalized by $\vc v^\top\vc 1=1$ and $\vc u^\top\vc v=1$, and let $\Delta_{v}=\diag(\vc v)$. We first rewrite $\exp(\Omega t)$ as $$\exp(\Omega t)=\exp(\Omega_0 t)\,\Delta_{v}\,\exp({\Theta} t)\,\Delta_{v}^{-1},$$ where $${{\Theta}}=\Delta_{v}^{-1}\Omega\Delta_{v}-\Omega_0 I.$$ It is easy to show that $\Theta$ satisfies all the properties to be the generator of a continuous-time Markov chain: $\Theta_{ij}\geq 0 $ for $i\neq j$ and $0\geq\Theta_{ii}=-\sum_{i\neq j} \Theta_{ij}$. This Markov chain is called the \textit{dual} of the branching process with mean population size matrix $\exp(\Omega t)$.
The integral in \eqref{repr} then becomes  \begin{eqnarray*}
    \lefteqn{\int_0^t \exp(\Omega \tau) \,\partial_p\Omega\,\exp(\Omega (t-\tau))\,d\tau=}\\ &&\exp(\Omega_0\,t)\,\Delta_{v}\int_0^t \exp({\Theta} \tau)\,\Delta_{v}^{-1}\partial_p\Omega\,\Delta_{v}\,\exp({\Theta} (t-\tau))d\tau\,\Delta_{v}^{-1}.\end{eqnarray*} Since $\exp(\Theta t)$ is now a probability transition matrix, we can use the uniformization method to solve the integral: let $c>\max_i|\Theta_{ii}|$ and $\hat{\Theta}=I+\frac{\Theta}{c}$.  We can write 
$$ \exp(\Theta t)=\sum_{k\geq 0} \exp(-ct)\,\frac{(ct)^k}{k!}\,\hat{\Theta}^k,$$ and since the matrix $\hat{\Theta} $ is substochastic, we have \begin{eqnarray*} ||\exp(\Theta t)||&\leq &\sum_{k\geq 0} \exp(-ct)\,\frac{(ct)^k}{k!}\,||\hat{\Theta}^k||\\&\leq &\sum_{k\geq 0} \exp(-ct)\,\frac{(ct)^k}{k!}.\end{eqnarray*} Let $K(t)$ be such that $\sum_{k= 0}^{K(t)} \exp(-ct)\,\frac{(ct)^k}{k!}\geq 1-\epsilon$ for a given $\epsilon>0$. Then,
$$ \exp(\Theta t)\approx \sum_{k=0}^{K(t)} \exp(-ct)\,\frac{(ct)^k}{k!}\,\hat{\Theta}^k,$$ and by replacing $\exp(\Theta \tau)$ and $\exp(\Theta( t-\tau))$ in the last integral by their approximations, we finally obtain the approximation \eqref{met2} for $\partial_p M(t)$. \end{proof}

\section{Interpretation of the elasticity} \label{apel}

In this section we interpret the elasticity in terms of a proportional response to a proportional perturbation.   Let $a$ be the elasticity of a measure $X$ of the model with respect to a parameter $p$, that is, $\Upsilon^X_p=\partial\log X/\partial\log p=a$. Assume that $p$ increases by $r\,\%$ ($r$ relatively small), and let $z$ be the induced proportional increase (or decrease, depending on the sign of $a$) of $X$, that is,$$p'=p(1+0.01r)\Rightarrow X'=X(1+z).$$ By taking the logarithm we obtain
  $$\log p'=\log p+\log(1+0.01r)\Rightarrow \log X'=\log X+ \log(1+z).$$ We thus have $$a=\frac{\partial\log X}{\partial\log p}\approx\frac{\log(1+z)}{\log(1+0.01r)},$$ and we obtain $$z\approx \exp(a\log(1+0.01r))-1.$$ So, for instance, if $p$ increases by $r=1\,\%$, then  then $X$ increases (or decreases) by approximately $100[\exp(a/100)-1]\,\%$. In Figure~\ref{ela}, we show the curve $y=100[\exp(a/100)-1]$ and we compare it with $y=a$. We see that the two curves are almost superimposed when $|a|$ is small.
  
  \begin{figure}[t!] 
	\begin{center}
	\includegraphics[angle=0,width=8cm]{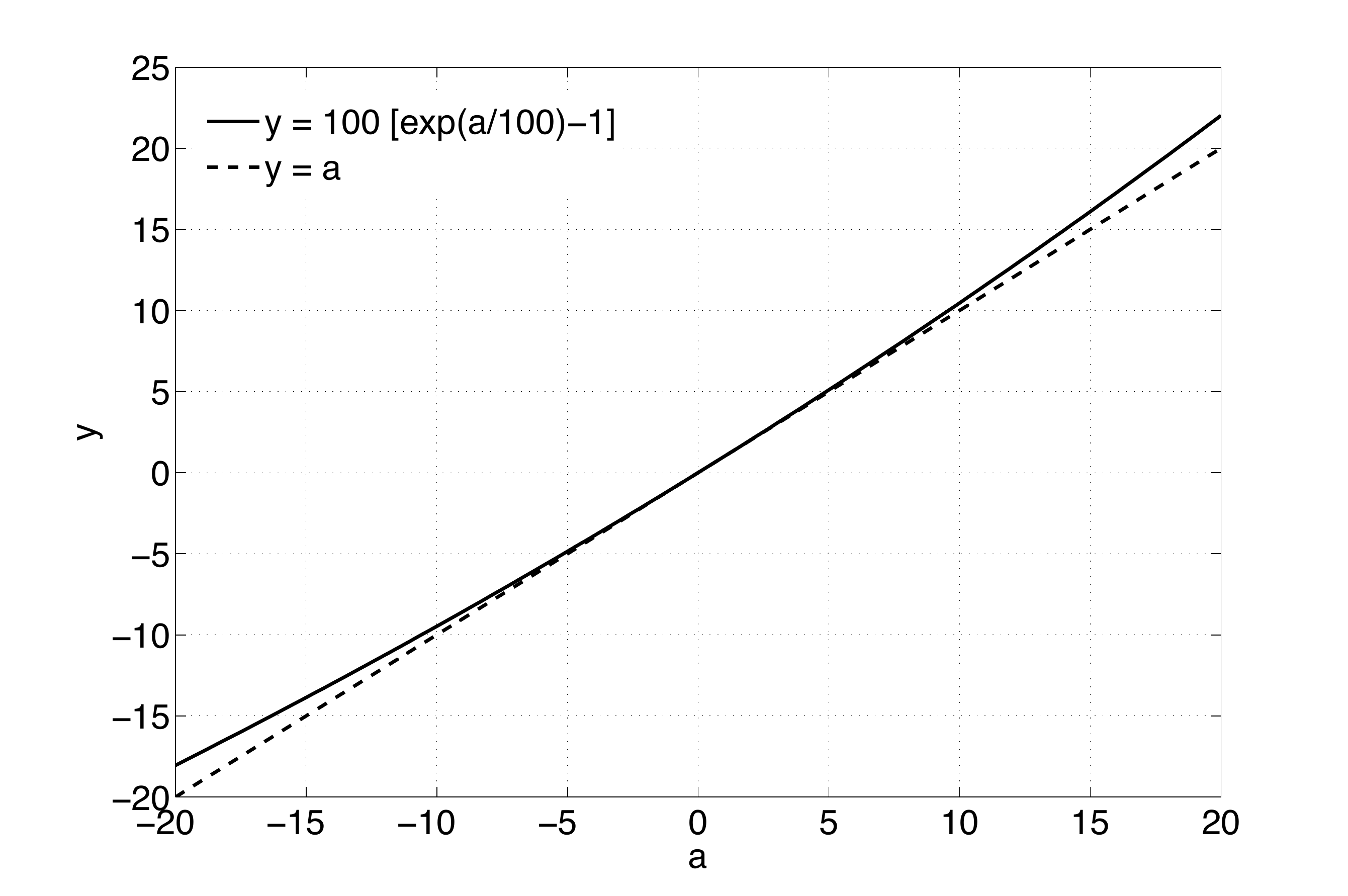} 
	\caption{\label{ela}\textbf{Elasticity. }Interpretation of the elasticity in terms of percent of increase or decrease of a measure of the model}
\end{center}
\end{figure}

\section{Validity of the branching process approximation}\label{valapprox}

In this section, we compare the mean instantaneous and cumulative population size obtained from the branching process defined in Section~\ref{bpn} to their analogue in the popular deterministic \textit{SIR} model. 

For that purpose, we illustrate the two measures $M(t)$ and $\vc D(t)$ on an example where there is a first infectious case in Orlando (city $i$) at time $t=0$. This city has a low transmission rate ($\beta_i=0.85$) but it is well connected with New York (city $j$, $T_{ij}=0.0033$) which has a high transmission rate ($\beta_j=1.1$). In Figure~\ref{fi2}, we show the growth of the population of infected individuals within the origin city of the disease, as well as in New York.  In that figure, we compare the branching process approximation to the 
SIR model $\{(S_i(t),I_i(t),R_i(t)), 1\leq i\leq 114\}$ which tracks, for every city $i$, the number of susceptible ($S_i(t)$), infected ($I_i(t)$), and removed ($R_i(t)$) individuals in city $i$ at time $t$,  and whose evolution is modeled by the nonlinear differential equations
\begin{eqnarray*}\dfrac{d S_i}{dt}&=& -\dfrac{\beta_i}{N_i} S_i \,I_i+\sum_{j\neq i} S_jT_{ji} -\sum_{j\neq i} S_iT_{ij} \\
\dfrac{d I_i}{dt}&=& \dfrac{\beta_i}{N_i} S_i \,I_i+\sum_{j\neq i}I_j T_{ji} - I_i\sum_{j\neq i} T_{ij}-d_i \,I_i\\
\dfrac{d R_i}{dt}&=& d_i\,I_i+\sum_{j\neq i}R_j T_{ji} -R_i\sum_{j\neq i} T_{ij},
\end{eqnarray*} where $(S_i(0), I_i(0),R_i(0))$ is specified, for $1\leq i\leq 114$.

We assume homogeneous mixing within each city. According to \cite{bd95}, during the course of a major epidemic, the disease grows like a branching process within each city $i$ until about $\sqrt{N_i}$ members of the population of $N_i$ individuals become infected. If we use this criterion to determine the maximum validity period of the branching process approximation in each city, we find that the approximation is accurate during the first $10$ days of the epidemic on average, the time value ranging from 6 days (Aspen, Colorado, $\sqrt{N_i}=81.6$, $\beta_i=1.1$) to 15 days (Houston, Texas,  $\sqrt{N_i}=2467.2$, $\beta_i=0.85$). 

This criterion is probably too strong for our purpose. If we rather fix to $1\%$ the relative error obtained when computing $M_{ii}(t)$ with the branching process versus $I_i(t)$ with the \textit{SIR} model, the approximation is accurate until $13$ days on average. The maximum approximation time depends on the metropolitan size of the city considered and on its associated transmission rate; it ranges from 5 days (Aspen) to 20 days (Dallas, Texas,  $\sqrt{N_i}=2554.8$, $\beta_i=0.85$). For Orlando ($\sqrt{N_i}=1473.4$) and New York ($\sqrt{N_i}=4366.9$), the two cities used to illustrate our measures, the approximation is accurate until respectively 19 and 15 days after the introduction of the first infectious case in the city. 

Figure~\ref{fi2} also shows that, while the epidemic initially started in Orlando, after a bit more than two weeks it mainly evolves in New York because of the relatively high travel rate between Orlando and New York and the high transmission rate in New York. This is in line with \cite{grais03} and confirms that air travel has a non-negligible effect on the spatial spread of the disease.

\begin{figure}[t] 
	\begin{center}
	\includegraphics[angle=0,width=10cm]{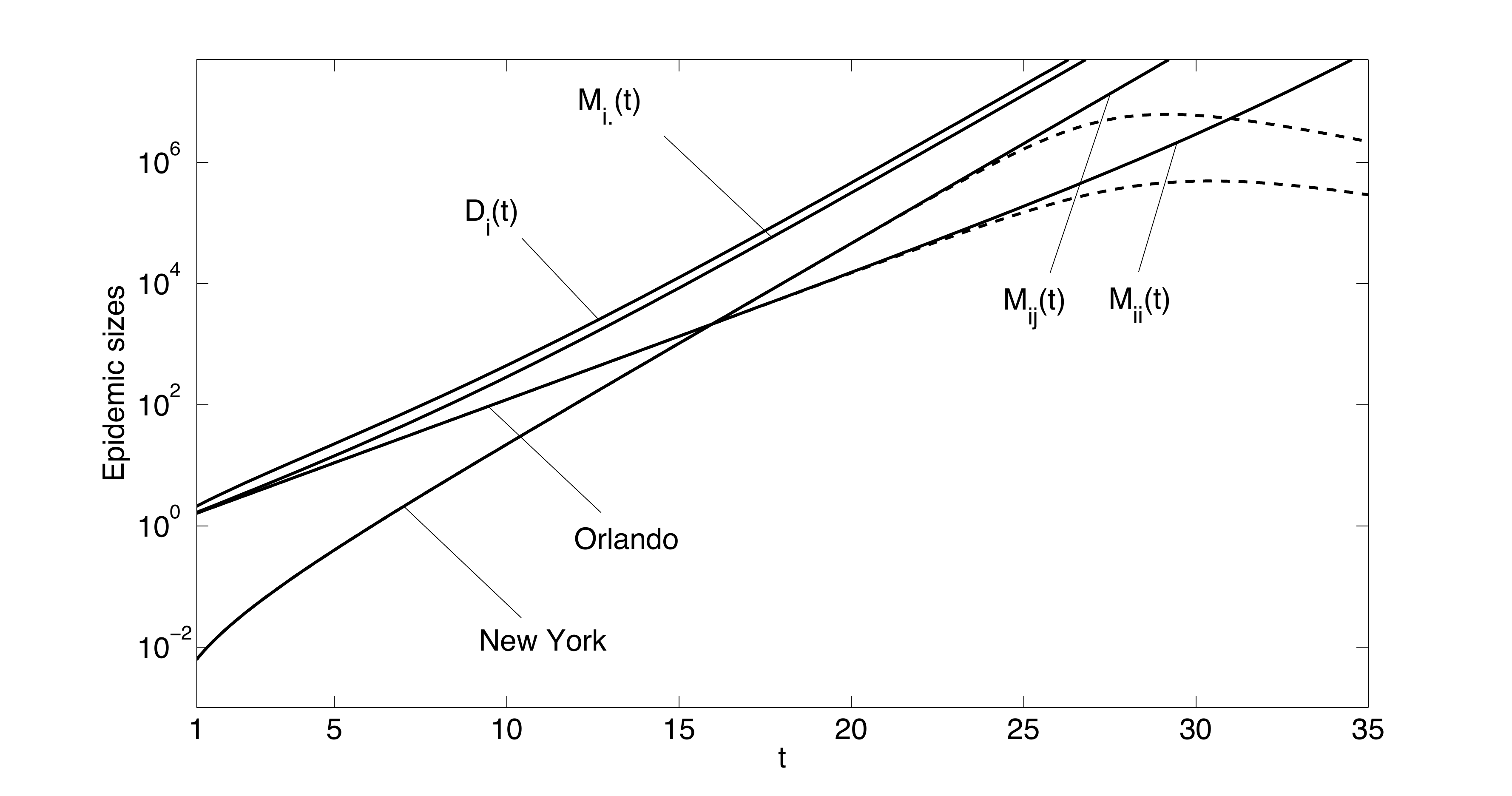} 
	\caption{\textbf{Mean instantaneous and cumulative epidemic sizes.} Mean epidemic sizes between the first and the 35th day (in logarithmic scale) if the disease is initiated by one infectious case in Orlando (city $i$) on day $t=0$. The plain lines correspond to the branching process approximation, the dashed lines correspond to the \textit{SIR} model. $M_{ii}(t)$ is the mean instantaneous number of infected individuals in Orlando, $M_{ij}(t)$ is the mean instantaneous number of infected individuals in New York (city $j$), $M_{i.}(t)=\sum_kM_{ik}(t)$ is the mean instantaneous total number of infected individuals in all cities, and $D_i(t)$
is the mean cumulative number of infected individuals in all cities. The branching process approximates the \emph{SIR} model reasonably during the first three weeks of the epidemic.
}
\label{fi2}
\end{center}
\end{figure}
\begin{rem} The expected instantaneous epidemic sizes $M_{ij}(t)$ obtained from the branching process approximation correspond to the solution of the linearisation of the quadratic differential equation for $I_j(t)$ around the disease-free equilibrium, that is, if we assume that $S_j\approx N_j$ for all $j$. Note that, thanks to its stochastic nature, the branching process approximation may also allow us to study other types of questions, such as the full distribution of the instantaneous and cumulative epidemic sizes during the early stages of the disease. These questions would be particularly relevant if the epidemic starts with a small number of infected individuals, as the strength of stochasticity increases as population sizes get smaller.
\end{rem}

\section{On-board transmission} \label{obt}

In this section, we show what the results of Sections~\ref{sedes} and~\ref{illu}  become when on-board transmission is taken into account. 

\subsection{Description of the model}\label{obtd}

To the best of our knowledge, there is not enough available data to conclude in an acceptable on-board transmission rate for influenza-like diseases. Therefore, for the sake of simplicity, we assume that during a flight an infected individual may infect up to two susceptible individuals (his/her two closest neighbours in the airplane), independently of each other. This is probably a lower bound of the real number of new infections, but it is sufficient for our purpose of showing the differences with the case without on-board transmission. 

We model the number of new infections generated by an infected individual during a flight from city $i$ to city $j$ with a binomial distribution $\textrm{B}(2,p_{ij}),$ where $p_{ij}$ is the probability that an infected individual infects another individual during a flight from city $i$ to city $j$. Since no real data are available for these on-board transmission probabilities, for the sake of our analysis we shall assume that they take the hypothetical form $$p_{ij}=\dfrac{5h_{ij}}{5h_{ij}+1},$$where $h_{ij}$ is the flight time (in hours) from city $i$ to city $j$, which is directly proportional to the distance between $i$ and $j$. The data on non-stop distances between cities are available from the US Department of Transportation \cite{DOT}. 
For $i\neq j$, the probability $(P_1)_{ij}$ that an infected individual infects exactly one susceptible individual during the flight between city $i$ and city $j$ is given by $(P_1)_{ij}=2p_{ij}(1-p_{ij})$. Similarly, the probability $(P_2)_{ij}$ that an infected individual infects exactly two susceptible individuals during the flight between city $i$ and city $j$ is given by $(P_2)_{ij}=p_{ij}^2$. These probabilities depend on the flight time and are shown in Figure~\ref{fiflight}.

The travel rate from city $i$ to city $j$ without subsequent transmission, denoted by $T'_{ij}$, or associated with the transmission of the disease to one or two individuals, denoted by $(C_1)_{ij}$ and $(C_2)_{ij}$ respectively, are given by
\begin{eqnarray*}T'_{ij}&=&T_{ij}\, [1-(P_1+P_2)_{ij}],\\(C_1)_{ij}&=&T_{ij}\,(P_1)_{ij},\\(C_2)_{ij}&=&T_{ij} \,(P_2)_{ij},\end{eqnarray*}for $i\neq j$.

In the branching process approximation, the travel rate matrix then becomes $T'$ (the diagonal of $T'$ being the same as that of $T$), and if we assume that air travel is instantaneous (that is, flight times are negligible), there may be up to two transmissions at a time. Therefore, the birth rate matrices are such that the only nonzero entries of the matrix $B'_1$ are $$(B'_1)_{i;ii}=\beta_i\quad\mbox{and}\quad (B'_1)_{i;jj}=(C_1)_{ij},$$ the only nonzero entries of the matrix $B'_2$ are $$(B'_2)_{i;jjj}=(C_2)_{ij},$$ and $$B'_k=0\quad\mbox{for }k\geq3.$$ With these, $B'_1(I\otimes \vc 1)=B'_1(\vc 1\otimes I)=\diag(\vc\beta)+C_1$ and $B'_2(I\otimes\vc 1\otimes \vc1)=B'_2(\vc 1\otimes I\otimes \vc1)=B'_2(\vc1\otimes\vc 1\otimes I )=C_2$,   so that the matrices $\Omega$ and $R$ become
\begin{eqnarray*}\Omega'&=&T'+2\,(\diag(\vc\beta)+C_1)+3\,C_2\\R'&=&-\left[I+T'^{-1}(\diag(\vc\beta)+C_1+C_2)\right]^{-1}\,T'^{-1}(\diag(\vc\beta)+C_1+2C_2).\end{eqnarray*}

\begin{figure}[t!] 
	\begin{center}
	\includegraphics[angle=0,width=8cm]{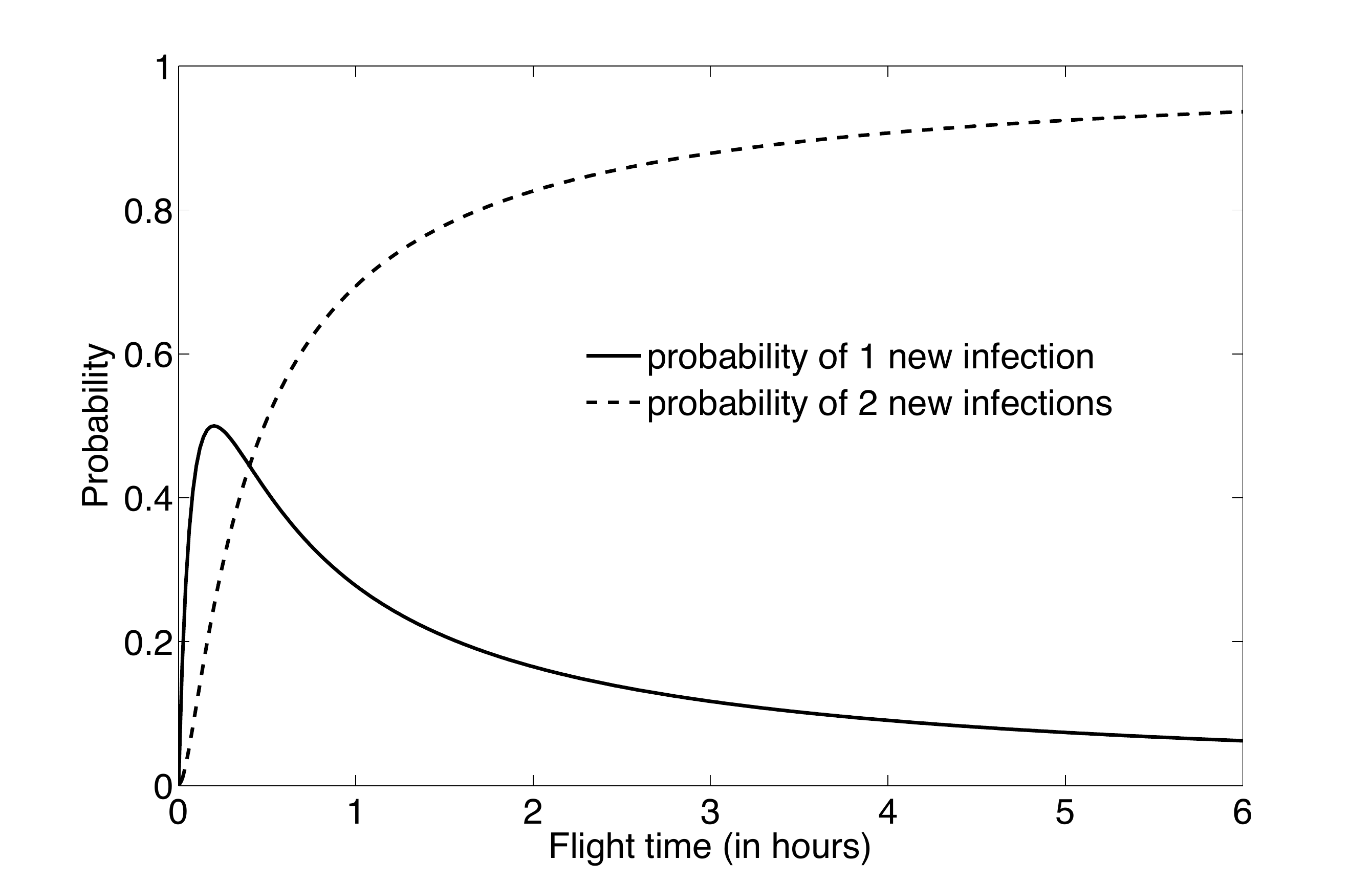} 
	\caption{\label{fiflight}\textbf{On-board transmission probabilities. }Probability that an infected individual infects one or two new individual(s) as a function of the flight time (in hours).}
\end{center}
\end{figure}

Note that our analysis could be adapted to more than two new infections during a flight. This would require to define subsequent matrices $P_i$ and $C_i$ for $i>2$, but would complicate the model unnecessarily for our purpose.

\subsection{Results and Discussion}\label{obtr}
We found that the difference in the global growth of the disease (that is, in $R_0$) is very small when taking on-board transmission into account.
However, at the city level, the difference can be more noticeable, as shown in Table~\ref{tab2b2}: for example, we see that if the disease starts in Orlando, the expected cumulative epidemic size after two weeks is more than twice larger with on-board transmission than without.

 \begin{table}[t]
\begin{center} \caption{\label{tab2b2} Mean cumulative epidemic sizes after 14 days associated with four origin cities if on-board transmission is taken into account. We indicate the ratio of the mean epidemic size with on-board transmission to the corresponding mean size without on-board transmission as  given in Table~\ref{tab2b}.}
\medskip
\begin{tabular}{lcc}
  City $j$ & $D_j(14)$ & Ratio
  \\ 
  \hline
  New York &$ 6.49\cdot10^4$ & $1.0945$\\
  Chicago &$6.70\cdot10^4$ & $1.1347$\\
  San Francisco  &$3.25\cdot10^4$&$1.5231$\\
  Orlando &$1.71\cdot10^4$ & 2.7003
\end{tabular}
\end{center}
\end{table}

In order to emphasize the effects of the on-board transmission mechanism on the elasticities, we shall compare elasticities without and with on-board transmission on the same graphs, even if for some cities the branching process approximation might be less accurate around $t=14$ days when on-board transmission is taken into account. The plain lines in the graphs correspond to the model with on-board transmission, and the dashed lines correspond to the case without on-board transmission.
  
\begin{figure}[t]\begin{center}
 
 \begin{minipage}{0.4\textwidth}
    \begin{figure}[H]
        \hspace*{-1.6cm}  \includegraphics[scale=0.21]{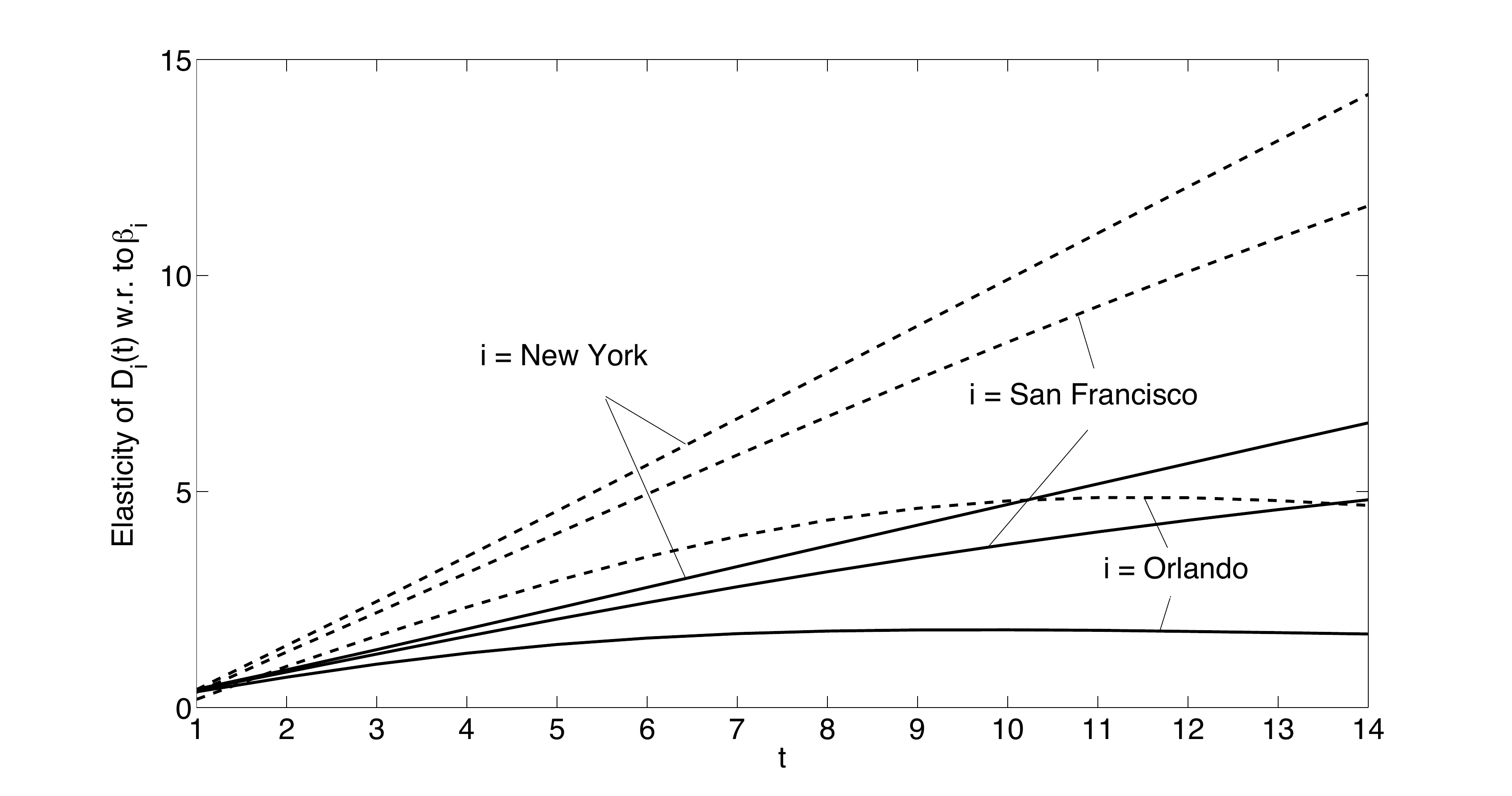}
        \vspace*{-0.5cm}\caption{\label{fielbet3bis} \textbf{Elasticity of the mean cumulative epidemic size with respect to the transmission rates. } 
	 Same graphs as in Figure~\ref{fielbet3}, with on-board transmission; the plain lines correspond to the model with on-board transmission, and the dashed lines correspond to the model without on-board transmission.
	        }
   \end{figure}
\end{minipage}
\hspace{3ex} 
\begin{minipage}{0.4\textwidth}
    \begin{figure}[H]
      \hspace*{-1.3cm}  \includegraphics[scale=0.21]{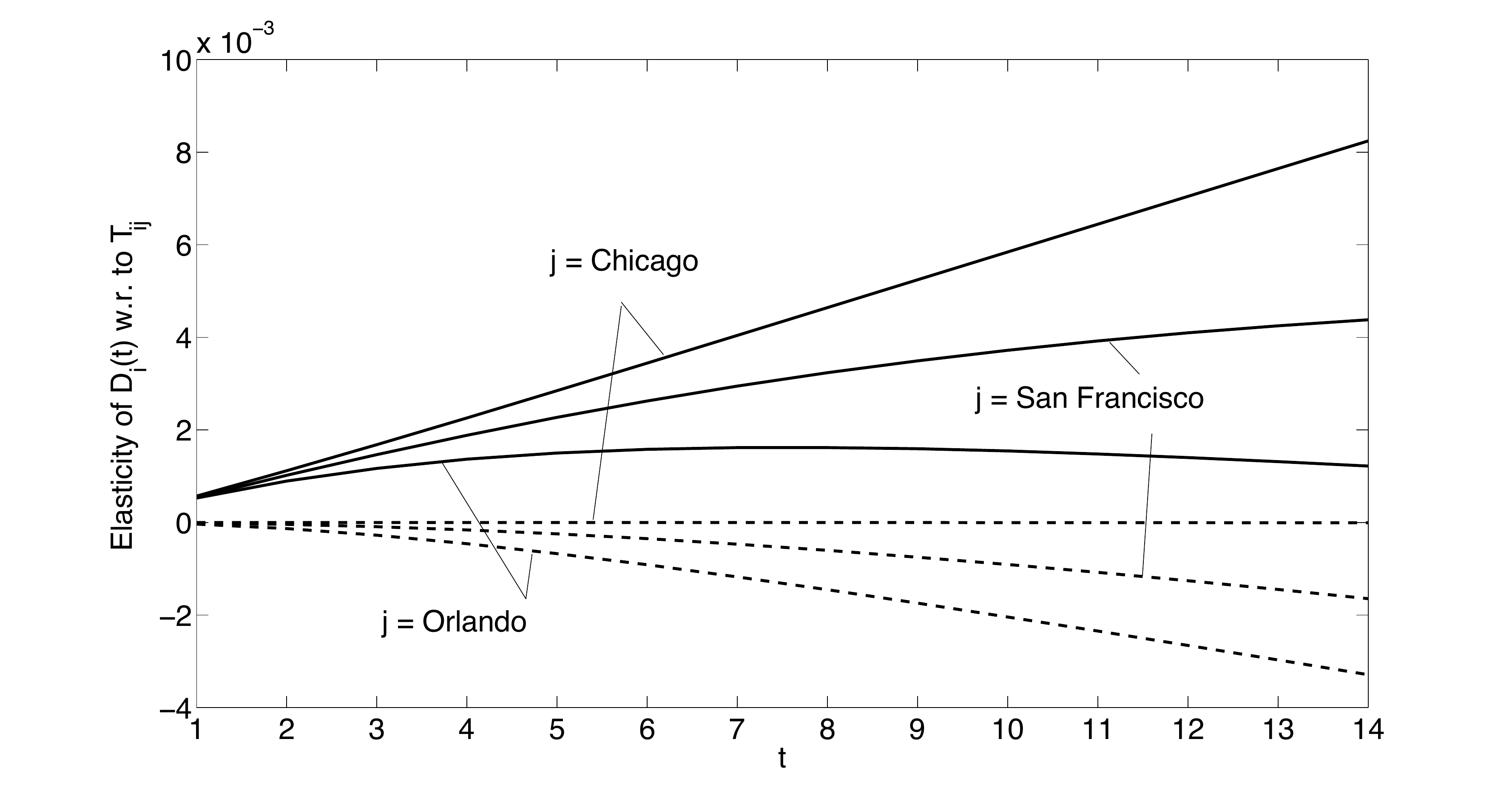}
        \vspace*{-0.5cm}\caption{\label{fielD01bisbis} 
	\textbf{Elasticity of the mean cumulative epidemic size with respect to the travel rates. }
	Same graphs as in Figure~\ref{fielD01bis}, with on-board transmission; the plain lines correspond to the model with on-board transmission, and the dashed lines correspond to the model without on-board transmission. }
       
    \end{figure}
\end{minipage}\vspace*{-0.5cm}
\end{center}
\end{figure}

 \paragraph*{Elasticity of the epidemic size with respect to the transmission rates.}
 
In Figure~\ref{fielbet3bis} we show a comparison of the elasticities of the mean cumulative size of the epidemic with respect to the transmission rates with and without on-board transmission.
We see that in general on-board transmission decreases markedly the elasticity of $D_i(t)$ with respect to $\beta_i$.

 \paragraph*{Elasticity of the epidemic size with respect to the travel rates.}

In Figure~\ref{fielD01bisbis}, we see that the elasticity of the mean cumulative size of the epidemic with respect to the travel rates is larger when on-board transmission is taken into account. This indicates, as expected, that in that case a change in the travel rates would affect much more the dynamics of the disease than when on-board transmission is not taken into account. 
 More precisely, we see that the elasticities start increasing to positive values with respect to the travel rates out of New York, due to the risk of transmission on board airplanes. After a few days, however, the elasticities reach a maximum value and start decreasing to negative values, which is especially clear for Orlando. This indicates that, in the long term and under our on-board transmission assumptions, it would still be beneficial from a sanitary point of view that infected people travel out of New York to cities where the transmission rate is lower.

\paragraph*{Elasticity of $R_0$.} Let $C=C_1+C_2$. The expression given in \eqref{dpR} for $\partial_p R$ simplifies to
\begin{eqnarray}\nonumber\partial_p R'&=&\left[I+T'^{-1}(\diag(\vc\beta)+C)\right]^{-1}\left[T'^{-1}(\partial_p T') T'^{-1}(\diag(\vc\beta)+C)-T'^{-1}\partial_p(\diag(\vc\beta)+C)\right]\cdot\\\label{ro2}&&\left\{\left[I+T'^{-1}(\diag(\vc\beta)+C)\right]^{-1}+I\right\}.\end{eqnarray}

The values of the elasticity of $R_0$ with respect to the different parameters with on-board transmission 
are identical to those obtained in the case without on-board transmission up to the $6$th decimal. We conclude that on-board transmision (as we define it) has a negligible effect on the sensitivity of $R_0$ with respect to these parameters.

\section{Vaccination}\label{apvac}

The vaccination problem is stated as follows: given a population and an epidemiological model, what is the smallest fraction of the population that needs to be vaccinated to prevent the epidemic to break out? And once this number is reached, how many new infections do we prevent with each additional vaccinated individual? We first answer these two questions in the scalar case (uniform vaccination) and the vector case (city-dependent vaccination) without considering on-board transmission. We then discuss what the results become when on-board transmission is taken into account.

\subsection{Uniform vaccination} \label{apvacu}
\paragraph*{Minimal fraction to be vaccinated.}
In our case, if we assume that vaccination is done uniformly over the population of each city, without consideration of age, social status, etc., then vaccination has the same effect as reducing the transmission rate. Let us suppose for now that the vaccination is done in the same proportion over the whole country, and that a fraction $r$ of the population is vaccinated. It leads to a new transmission rate vector $\vc \beta_v = (1-r)\vc \beta$. In the sequel, we write the subscript $v$ each time a quantity depends on this new vector. The minimal fraction of the population to be vaccinated in order to prevent a breakout of the epidemic is then given by the minimization problem:
\begin{equation}\label{opt1}
\begin{aligned}
& \underset{r}{\text{minimize}}
& & r\sum\limits_{i}{N_i}\\
& \text{subject to}
&&\lambda_{max}(\Omega_v)\leq 0.
\end{aligned}
\end{equation}

If we write, for simplicity, $\Delta_{\beta}=\diag(\vc \beta)$, one has $\Omega = T+ 2\Delta_{\beta}$ and $\Omega_v = T+ (2-r)\Delta_{\beta}$.
Since $\Omega_v$ is not symmetric (it results from the sum of $T$, which is not symmetric, and a diagonal correction $(2-r)\Delta_{\beta}$), there is no guarantee that its eigenvalues are real. However, since $T=N^{-1}A$, where $N$ is $\diag(N_i)$ the diagonal matrix of metropolitan populations, and $A$ is a symmetric air travel matrix (as in \cite{grais03}),
 the matrix $N^{1/2}TN^{-1/2}$ is also symmetric and has the same eigenvalues as $T$. Consequently, $\Omega_v$ has the same eigenvalues as $N^{1/2}\Omega_v N^{-1/2}$ which is symmetric too. For symmetric matrices, the condition $\lambda_{max}(\Omega_v)\leq 0$ can be rewritten using Rayleigh quotients as 
$$\frac{\vc v^\top N^{1/2}\Omega_vN^{-1/2}\vc v}{\vc v^\top\vc v}\leq 0, \qquad \forall \vc v:\|\vc v\|^2>0.$$
By expanding the expression of $\Omega_v$ and simplifying, one obtains
$$\vc v^\top N^{1/2}TN^{-1/2}\vc v\leq (r-2)\vc v^\top \Delta_{\beta}\vc v, \qquad \forall \vc v:\|\vc v\|^2>0.$$
If all $\beta_i$ are nonzero, then $\Delta_{\beta}$ is invertible, and we can write $\vc v = \Delta_{\beta}^{-1/2}\vc w$, so that the condition becomes
$$\vc w^\top\Delta_{\beta}^{-1/2}N^{1/2}TN^{-1/2}\Delta_{\beta}^{-1/2}\vc w\leq (r-2)\vc w^\top \vc w,\qquad \forall \vc w: \|\vc w\|^2 >0.$$
By dividing both sides by $\vc w^\top\vc w$, one obtains
$$\frac{\vc w^\top N^{1/2}\Delta_{\beta}^{-1/2}T\Delta_{\beta}^{-1/2}N^{-1/2}\vc w}{\vc w^\top\vc w}\leq (r-2),\qquad \forall \vc w: \|\vc w\|^2 >0,$$
which can finally be translated into 
\begin{equation}\label{eq:scalcond}2+\lambda_{max}(\Delta_{\beta}^{-1}T)\leq r.\end{equation}
In our case, we conclude that $r$ should be larger than $0.6917$, so at least $69.17\%$ of the population has to be vaccinated in order to prevent the epidemic to break out.


%

\subsection{City-dependent vaccination}\label{apvaccd}

\paragraph*{Minimal fraction to be vaccinated.}
When vaccination is not done uniformly over the population but varies from city to city, the vaccination ratio is described by a vector $\vc r$, where $r_i$ is the fraction of the population of city $i$ that is vaccinated. If we define $V=\diag(\vc r)$, the optimisation problem becomes
\begin{equation}\label{opt2}
\begin{aligned}
& \underset{\vc r}{\text{minimize}}
& & \sum\limits_{i}{r_iN_i}\\
& \text{subject to}
&&\lambda_{max}(\Omega_v)\leq 0,
\end{aligned}
\end{equation}
with $\Omega_v = T + (2I-V)\Delta_{\beta}$. Again, this matrix is not symmetric but it has the same eigenvalues as $N^{1/2}\Omega_vN^{-1/2}$ which is symmetric.
The constraint $\lambda_{max}(\Omega_v)\leq 0$ can then be written as
$$\vc v^\top N^{1/2}(T+(2I-V)\Delta_{\beta})N^{-1/2}\vc v\leq 0 \qquad\forall \vc v:\|\vc v\|^2 > 0.$$
By writing again $\vc w=\Delta_{\beta}^{1/2}\vc v$, one obtains after some elementary manipulations
\begin{equation}\label{eq:sdp}
\frac{\vc w^\top [V - N^{1/2}(\Delta_{\beta}^{-1/2}T\Delta_{\beta}^{-1/2}+2I)N^{-1/2}]\vc w}{\vc w^\top\vc w} \geq 0 \qquad\forall \vc w:\|\vc w\|^2 >0,
\end{equation}
which simply means that $V-N^{1/2}(\Delta_{\beta}^{-1/2}T\Delta_{\beta}^{-1/2}+2I)N^{-1/2}$ has to be positive semi-definite.

In general, the solution of such an optimization problem is not trivial, but it can be solved using \textit{Linear Matrix Inequality} (\textit{LMI}) solvers. However, in this case, we can derive a fairly accurate approximation. 
Recall that the matrix $T$ contains on its off-diagonal elements the travel rates between cities, while the diagonal elements are given by 
$T_{ii} = -\beta_i -d_i -\sum\limits_{j\neq i}{T_{ij}}.$ In the present situation, the travel rates are small with respect to transmission and recovery rates, and we can approximate $T\approx\diag(-\vc\beta -\vc d)$. Then, the condition becomes that $V-\Delta_{\beta}^{-1}T-2I$ is positive semi-definite. Since we want to minimize a weighted sum of the elements of $V$, with only positive weights, this reduces to solve $r_i - \beta_i^{-1}(-\beta_i-d_i) - 2 \geq 0$ for each $i$, which immediately gives
\begin{equation}\label{eq:simplified}
r_i\geq1-(d_i/\beta_i).
\end{equation} 

We solved the original problem with condition (\ref{eq:sdp}), using the Matlab software \texttt{CVX} \cite{gb08,cvx} for dealing with the \textit{LMI} condition, and we denote the solution by $\vc{r}_{cvx}$. The values of the vector $\vc{r}_{cvx}$ range between $0.5965$ and $0.6929$, with a mean of $0.6574$.
We compared this solution with the solution obtained in (\ref{eq:simplified}), denoted by $\vc{r}_{approx}$. In all cities the optimal vaccination fractions obtained with both methods are very similar: the largest relative error between the entries of $\vc{r}_{cvx}$ and $\vc{r}_{approx}$ is 1\%. There is a clear correlation between the intensity of the travel rates and the relative error; the approximation is excellent in weakly connected cities and less accurate in the most connected cities.

However, this approximation is valid in our case because travel rates are very small with respect to the transmission and removal rates $\vc\beta$ and $\vc d$. This is no longer the case when travel has a larger impact than here, that is, when commuting traffic is added, or if on-board transmission is taken into account, as shown in Appendix~\ref{apvacobt}.

\paragraph*{Sensitivity of $\vc D$ with respect to $\vc r$.}

Now, the sensitivity of vector $\vc D$ with respect to $r_i$ (the other $r_j$ for $j\neq i$ being held constant) is given by $$\frac{\partial \vc D}{\partial r_i} =(-\Omega_v)^{-1}\,\frac{\partial \Omega_v}{\partial r_i}\,(-\Omega_v)^{-1}\,\vc d,$$ where the only nonzero entry of the matrix $\partial \Omega_v/\partial r_i$ is the entry $(i,i)$ equal to $-\beta_i$, and $\partial r_i = 1/N_i$.
In this case, when the proportion $\vc r$ of vaccinated individuals in each city is large enough to prevent an epidemic to outbreak, the effect of an additional vaccine in city $i$ is then given by
$$\partial \vc D =\frac{1}{N_i}(-\Omega_v)^{-1}\,\frac{\partial \Omega_v}{\partial r_i}\,(-\Omega_v)^{-1}\,\vc d.$$

Similar to the uniform case, the value of $\partial D_i$ decreases rapidly when $\vc r$ moves away from $\vc r_c$, as
shown in Figure~\ref{fivac2}, which depicts the number of people who would escape from the disease if we introduce one additional vaccine {in the origin city of the disease}, as a function of $(\vc r-\vc r_c)_i\times N_i$, the initial additional number of vaccinated people with respect to the critical number in the city, $(\vc r_c)_i N_i$. 
\begin{figure}[t!] 
	\begin{center}
	\includegraphics[angle=0,width=9cm]{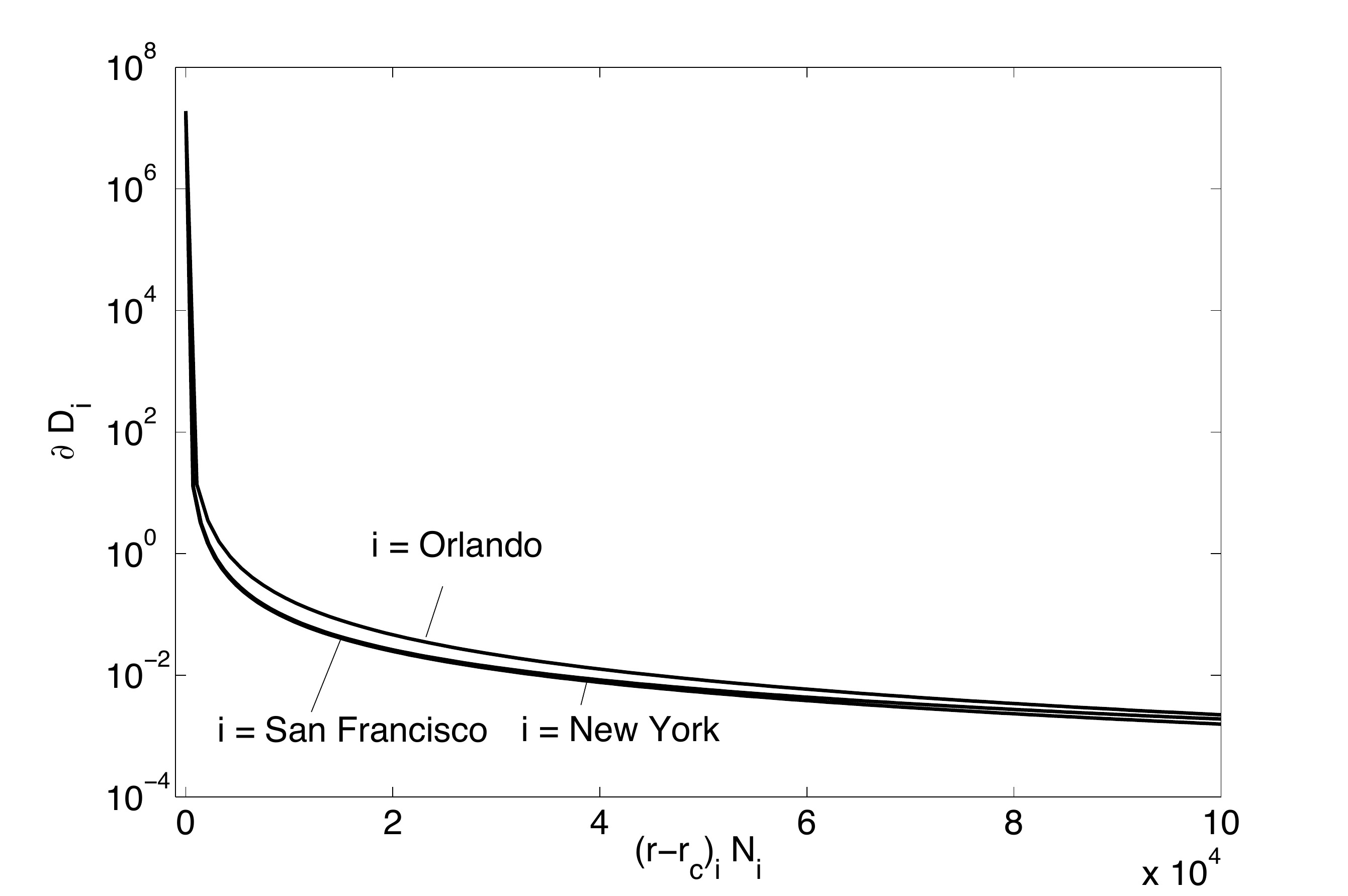} 
	\caption{\label{fivac2}\textbf{City-dependent vaccination. }Number of prevented infections per additional vaccine in the origin city of the disease as a function of the difference between the initial number of vaccinated people in the city ($r_i N_i$) and the critical vaccinated population in the city ($(\vc r_c)_i N_i$), for three origin cities of the disease.}
\end{center}
\end{figure}

\begin{rem}If we assume that the vaccination campaign only applies to the 114 cities considered, then the results in this section show that, in order to prevent the epidemic to break out, at least $r_c\,\sum_iN_i=1.64\cdot10^8$ individuals have to be vaccinated in the uniform case, as opposed to $\sum_i (\vc r_c)_i\,N_i=1.57\cdot 10^8$ individuals in the city-dependent case. This highlights the benefit of doing city-dependent vaccination, since approximately seven million vaccinations less are needed to avoid an epidemic outbreak.\end{rem}

\subsection{Effect of on-board transmission}\label{apvacobt}
The optimal vaccination fraction can be estimated in a similar way when on-board transmission is taken into account.
The optimization problems \eqref{opt1} and \eqref{opt2} stay unchanged, except for the expression of $\Omega_v$, which is now given by $$\Omega_v=T'+(2I-V)\Delta_\beta+2C_1+3C_2.$$ For simplification reasons, we write $T_1 = T'+2C_1+3C_2,$ so that $\Omega_v=T_1+(2I-V)\Delta_\beta.$

Note that similar to $T$, $T_1$ is not symmetric, but can be written as the product of $N^{-1}$ and a symmetric matrix, so that $T_1$ has the same eigenvalues as $N^{1/2}T_1N^{-1/2}$, which is a symmetric matrix (the same holding for $T'$, $C_1$ and $C_2$).

\paragraph*{Uniform vaccination.}
The same manipulation as in the case without on-board transmission provides the condition equivalent to (\ref{eq:scalcond}): $\lambda_{max}(\Omega_v)\leq 0$ if and only if
\begin{equation}
2+\lambda_{max}(\Delta_\beta^{-1}T_1) \leq r.
\end{equation}
In our case, $r$ should be larger than 0.6960, which is a bit larger than the minimal vaccination fraction without on-board transmission; this shows again the influence of on-board transmissions on the size of the epidemic.

\paragraph*{City-dependent vaccination.}
Similar to the city-dependent case without on-board transmission, the condition $\lambda_{max}(\Omega_v)\leq 0$ of the optimisation problem \eqref{opt2} can be finally written as that $$V - N^{1/2}(\Delta_\beta^{-1/2}T_1\Delta_\beta^{-1/2} + 2I)N^{-1/2}$$ has to be positive semi-definite. 

The optimisation problem can be solved using \texttt{CVX} and be compared with the simplified solution $\vc{r}_{approx}$ obtained when assuming that the travel rates are negligible. The obtained vaccination rates $\vc{r}_{cvx}$ are higher than in the case without on-board transmission: their values range between $0.6015$ and $0.7824$ with a mean of $0.6682$. Strikingly, since the influence of air travel is enhanced by the on-board transmissions, we observe larger differences between the solution $\vc{r}_{cvx}$ returned by \texttt{CVX} and the simplified solution $\vc{r}_{approx}$. In the present case, 55 cities have a relative error between $\vc{r}_{cvx}$ and $\vc{r}_{approx}$ larger than 1\%  (the largest relative error is 12\%), and there is a very clear trend between the connectivity of a city and the error; the most connected cities have a large relative error, while the most isolated cities have a well-approximated vaccination rate.

\section{Complete data}\label{apcd}

Table~\ref{fig:data} provides the complete list of the 114 American cities considered in this paper, with their corresponding metropolitan population (2011 estimates of the United States Census Bureau), and automn-winter transmission rate. The four cities which illustrate our sensitivity analysis in Section~\ref{illu} are in bold. 

The $114\times114$ travel rate matrix is too large to be represented as a whole and can be computed according to \eqref{tr} using the average daily number of passengers for each city pair obtained from the Domestic Airline Fares Consumer Report of the US Department of Transportation for the first Quarter of 2011 \cite{DOT}.

\begin{table}[t]\scriptsize
\begin{center}\caption{\label{fig:data}\normalsize{The 114 cities considered, together with their metropolitan population $N_i$ (based on the 2011 estimates of the United States Census Bureau), and their winter transmission rate $\beta_i$ (based on \cite{grais03}).} }
\medskip
\tabcolsep=0.13cm\begin{tabular}{lll|lll}
City $i$	&	$N_i	$&	$\beta_i$	&City $i$	&	$N_i	$&	$\beta_i$	\\\hline
	&		&		&			&&			\\
Albany	&	857592	&	1.1	&	Little Rock	&	685488	&	1.02	\\
Albuquerque	&	857903	&	1.02	&	Los Angeles Metro Area 4/	&	12870000	&	1.02	\\
Allentown/Bethlehem	&	816012	&	1.1	&	Louisville	&	1259000	&	1.02	\\
Amarillo	&	246474	&	0.85	&	Lubbock	&	276659	&	0.85	\\
Aspen; CO (urban)	&	6658	&	1.1	&	Madison	&	570025	&	1.1	\\
Atlanta	&	5475000	&	1.02	&	Medford	&	201286	&	1.02	\\
Atlantic City	&	271712	&	1.02	&	Memphis	&	1305000	&	1.02	\\
Austin	&	1705000	&	0.85	&	Miami/Ft. Lauderdale 4/	&	5547000	&	0.85	\\
Baltimore/Washington 4/	&	2691000	&	1.02	&	Midland/Odessa	&	266941	&	0.85	\\
Baton Rouge; LA	&	786947	&	0.85	&	Milwaukee	&	1560000	&	1.1	\\
Bellingham	&	200434	&	1.1	&	Minneapolis	&	3270000	&	1.1	\\
Billings; MT	&	154553	&	1.1	&	Mission/McAllen/Edinburg	&	741152	&	0.85	\\
Birmingham	&	1131000	&	1.02	&	Moline	&	379066	&	1.1	\\
Bloomington/Normal	&	167699	&	1.1	&	Myrtle Beach	&	263868	&	1.02	\\
Boise	&	606376	&	1.1	&	Nashville	&	1582000	&	1.02	\\
Boston/Providence 4/	&	6190000	&	1.1	&	New Orleans	&	1190000	&	0.85	\\
Buffalo	&	1124000	&	1.1	&	\textbf{New York Metro Area 4/	}&	\textbf{19070000}	&	\textbf{1.1}	\\
Burlington	&	208055	&	1.1	&	Newburgh/Poughkeepsie	&	677094	&	1.1	\\
Cedar Rapids/Iowa City; IA	&	450462	&	1.1	&	Newport News/Williamsburg	&	1674000	&	1.02	\\
Charleston	&	659191	&	0.85	&	Norfolk 	&	1647000	&	1.02	\\
Charlotte	&	1746000	&	1.02	&	Oklahoma City	&	1227000	&	1.02	\\
\textbf{Chicago Metro Area 4/}	&	\textbf{9581000}	&	\textbf{1.1}	&	Omaha	&	849517	&	1.1	\\
Cincinnati; KY	&	2172000	&	1.02	&	\textbf{Orlando}	&	\textbf{2082000}	&	\textbf{0.85	}\\
Cleveland/Akron 4/	&	2790935	&	1.1	&	Palm Springs; CA	&	4143000	&	0.85	\\
Colorado Springs	&	626227	&	1.1	&	Panama City; FL	&	164767	&	0.85	\\
Columbia; SC	&	744730	&	0.85	&	Pasco/Kennewick/Richland; WA	&	245649	&	1.1	\\
Columbus	&	1802000	&	1.1	&	Pensacola; FL	&	455102	&	0.85	\\
Corpus Christi	&	416095	&	0.85	&	Philadelphia	&	5968000	&	1.1	\\
Dallas/Fort Worth 4/	&	6448000	&	0.85	&	Phoenix	&	4364000	&	0.85	\\
Dayton	&	835063	&	1.1	&	Pittsburgh	&	2355000	&	1.1	\\
Denver	&	2552000	&	1.1	&	Plattsburgh; NY	&	81618	&	1.1	\\
Des Moines	&	562906	&	1.1	&	Portland	&	2242000	&	1.02	\\
Detroit	&	4403000	&	1.1	&	Raleigh/Durham	&	1627228	&	1.02	\\
Eagle; CO	&	61699	&	1.1	&	Reno	&	419261	&	1.1	\\
El Paso	&	751296	&	0.85	&	Richmond	&	1238000	&	1.02	\\
Eugene; OR	&	351109	&	1.02	&	Rochester	&	1036000	&	1.1	\\
Fargo; ND	&	200102	&	1.1	&	Sacramento	&	2127000	&	1.02	\\
Fayetteville; AR	&	464623	&	0.85	&	Salt Lake City	&	1130000	&	1.1	\\
Flint	&	424043	&	1.1	&	San Antonio	&	2072000	&	0.85	\\
Fort Myers	&	586908	&	0.85	&	San Diego	&	3054000	&	1.02	\\
Fresno; CA	&	915267	&	0.85	&	\textbf{San Francisco/Oakland 4/}	&	\textbf{4318000}	&	\textbf{1.02}	\\
Grand Rapids	&	778009	&	1.1	&	{Santa Barbara; CA}	&	\textbf{407057}	&	\textbf{0.85}	\\
Greensboro/High Point	&	714765	&	1.02	&	Santa Rosa; CA	&	472102	&	0.85	\\
Harlingen/San Benito	&	396371	&	0.85	&	Sarasota/Bradenton	&	688126	&	0.85	\\
Harrisburg	&	536919	&	1.1	&	Savannah; GA	&	343092	&	0.85	\\
Hartford	&	1196000	&	1.1	&	Seattle	&	3408000	&	1.1	\\
Houston	&	5867000	&	0.85	&	Sioux Falls; SD	&	238122	&	1.1	\\
Huntsville	&	406316	&	1.02	&	Spokane	&	468684	&	1.1	\\
Indianapolis	&	1744000	&	1.1	&	St. Louis	&	2829000	&	1.02	\\
Islip	&	19070000	&	1.02	&	Syracuse	&	646084	&	1.1	\\
Jackson/Vicksburg	&	589041	&	1.02	&	Tallahassee; FL	&	360013	&	0.85	\\
Jacksonville	&	1328000	&	0.85	&	Tampa	&	2747000	&	0.85	\\
Kansas City	&	2068000	&	1.02	&	Tucson	&	1020000	&	1.02	\\
Key West; FL	&	73165	&	0.85	&	Tulsa	&	929015	&	1.02	\\
Knoxville	&	699247	&	1.02	&	West Palm Beach/Palm Beach	&	5547000	&	0.85	\\
Las Vegas	&	1903000	&	0.85	&	White Plains	&	19070000	&	1.1	\\
Lexington	&	470849	&	1.02	&	Wichita	&	612683	&	1.02	\\
\end{tabular}\end{center}\end{table}
\normalsize







%

\end{document}